\documentclass[journal,draftclsnofoot,onecolumn,12pt]{IEEEtran}

\newif\ifarxiv
\arxivtrue % set to true for arXiv version, false for conference version

\newif\ifonecolumn
\onecolumntrue % Uncomment this line when using one column format

\IEEEoverridecommandlockouts
\usepackage{multirow}
\usepackage{dsfont}
\usepackage{cite}
\usepackage{subcaption}
\usepackage{amsmath,amssymb,amsfonts, mathtools}
\usepackage{float}
\usepackage{makecell}
\usepackage{caption}
\usepackage{graphicx}
\usepackage{textcomp}
\usepackage{algorithm}
\usepackage[left=0.625in,right=0.625in,bottom=0.95in,top=0.75in]{geometry}
\usepackage{amsthm}
\usepackage{algpseudocode}
\usepackage{adjustbox}
\usepackage{graphicx}
\usepackage[table,xcdraw]{xcolor}
\usepackage{verbatim}
\usepackage{paralist}
\usepackage{array}
\newcolumntype{?}{!{\vrule width 1.1pt}}
\usepackage{tablefootnote}

\newtheorem{proposition}[]{Proposition}

\theoremstyle{remark}

\definecolor{green}{rgb}{0.0, 0.5, 0.0} % This is a dark shade of green
\algrenewcommand\algorithmiccomment[1]{\hfill\textcolor{gray}{\# #1}}
\allowdisplaybreaks
\usepackage{acro}
\newcolumntype{?}{!{\vrule width 1pt}}
\DeclareAcronym{ea}{
	short = EA,
	long = evolutionary algorithm,
	long-plural-form = {evolutionary algorithms}}
	
\DeclareAcronym{snr}{
  short = SNR,
  long = signal-to-noise ratio,}
  \DeclareAcronym{pdf}{
	short = PDF,
	long = probability density function,}
  \DeclareAcronym{sar}{
  short = SAR,
  long = Synthetic aperture radar,}

\DeclareAcronym{insar}{
  short = InSAR,
  long = Interferometric Synthetic Aperture Radar,}

\DeclareAcronym{ao}{
	short = {AO},
	long = {alternating optimization},
	long-plural-form = {alternating optimizations}
}
\DeclareAcronym{mimo}{
	short = {MIMO},
	long = {multiple-input multiple-output}
}

\DeclareAcronym{uav}{
        short = {UAV},
        long = {unmanned aerial vehicle},
        long-plural-form = {unmanned aerial vehicles}
}

\DeclareAcronym{fdma}{
	short = {FDMA},
	long = {frequency-division multiple-access},
}
\DeclareAcronym{1d}{
	short = {1D},
	long = {one-dimensional},
}
\DeclareAcronym{islr}{
	short = {ISLR},
	long = {integrated sidelobe ratio},
}

\DeclareAcronym{pslr}{
	short = {PSLR},
	long = {peak sidelobe ratio},
}

\DeclareAcronym{3d}{
        short = {3D},
        long = {three-dimensional},
}
\DeclareAcronym{pso}{
	short = {PSO},
	long = {particle swarm optimization},
}

\DeclareAcronym{2d}{
        short = {2D},
        long = {two-dimensional},
}

\DeclareAcronym{dem}{
        short = {DEM},
        long = {digital elevation model},
}
\DeclareAcronym{moop}{
	short = {MOOP},
	long = {multi-objective optimization problem},
}

\DeclareAcronym{gs}{
        short = {GS},
        long = {ground station},
        long-plural-form = {ground stations}
}

\DeclareAcronym{los}{
        short = {LOS},
        long = {line-of-sight},
}

\DeclareAcronym{sca}{
        short = {SCA},
        long = {successive convex approximation},
}

\DeclareAcronym{nesz}{
        short = {NESZ},
        long = {noise equivalent sigma zero},
}

\DeclareAcronym{wrt}{
        short = {w.r.t.},
        long = {with respect to },
}
\DeclareAcronym{rhs}{
        short = {r.h.s},
        long = {right-hand side },
}
\DeclareAcronym{gmti}{
	short = {GMTI},
	long = {ground moving target indication},
}
\DeclareAcronym{lhs}{
        short = {l.h.s},
        long = {left-hand side },
}
\DeclareAcronym{bcd}{
	short = {BCD},
	long = {block coordinate descent},
}
\DeclareAcronym{hoa}{
        short = {HoA},
        long = {height of ambiguity},
}
\DeclareAcronym{ga}{
	short = {GA},
	long = {genetic algorithm},
	long-plural-form = {genetic algorithms},
}

\DeclareAcronym{cga}{
	short = {CGA},
	long = {Continuous Genetic Algorithm},
	long-plural-form = {continuous genetic algorithms},
}

\DeclareAcronym{sa}{
	short = {SA},
	long = {simulated annealing},
}
\DeclareAcronym{genocop}{
	short = {Genocop II},
	long = {Genetic Algorithm for Numerical Optimization of Constrained Problems},
}
\DeclareAcronym{drl}{
	short = {DRL},
	long = {deep reinforcement learning},
}

\DeclareAcronym{dqn}{
	short = {DQN},
	long = {deep Q networks},
}

\DeclareAcronym{ddpg}{
	short = {DDPG},
	long = {deep deterministic policy gradient},
}

\begin{document}
	\begin{comment}Mohamed-Amine~Lahmeri and Robert Schober are with the  Institute for Digital Communications and  Martin Vossiek, Víctor Mustieles-Pérez, and Gerhard Krieger are with the Institute of Microwaves and Photonics (LHFT), Friedrich-Alexander University Erlangen-N\"urnberg (FAU), Germany (email:\{amine.lahmeri, robert.schober, martin.vossiek, victor.mustieles, gerhard.krieger\}@fau.de). Víctor Mustieles-Pérez and Gerhard Krieger are also with  the Microwaves and Radar Institute, German Aerospace Center (DLR), Weßling, Germany.
	\end{comment}
	\title{ Sensing Accuracy Optimization for Multi-UAV SAR Interferometry with Data Offloading \\\vspace{-2mm}
		\thanks{  This paper {was} presented in part at the IEEE International Conference on Communication,  Montréal, Canada, Jun. 2025 \cite{amine5}. This work was supported in part by the Deutsche Forschungsgemeinschaft (DFG, German Research Foundation) GRK 2680 – Project-ID 437847244.}
	}
	\ifonecolumn
	\author{\IEEEauthorblockN{Mohamed-Amine~Lahmeri\IEEEauthorrefmark{1},     Pouya Fakharizadeh\IEEEauthorrefmark{1}, Víctor Mustieles-Pérez\IEEEauthorrefmark{1}\IEEEauthorrefmark{2}, Martin Vossiek\IEEEauthorrefmark{1}, Gerhard Krieger\IEEEauthorrefmark{1}\IEEEauthorrefmark{2}, and
			Robert Schober\IEEEauthorrefmark{1}}\\ \vspace{-2mm}
		\IEEEauthorblockA{\IEEEauthorrefmark{1}Friedrich-Alexander-Universit\"at Erlangen-N\"urnberg (FAU), Germany\\
			\IEEEauthorrefmark{2}German Aerospace Center (DLR),  Microwaves and Radar Institute, Weßling, Germany\\
			\vspace{-8mm}}}
		\else
			\author{\IEEEauthorblockN{Mohamed-Amine~Lahmeri\IEEEauthorrefmark{1},     Pouya Fakharizadeh\IEEEauthorrefmark{1}, Víctor Mustieles-Pérez\IEEEauthorrefmark{1}\IEEEauthorrefmark{2}, Martin Vossiek\IEEEauthorrefmark{1}, \\ Gerhard Krieger\IEEEauthorrefmark{1}\IEEEauthorrefmark{2}, and
				Robert Schober\IEEEauthorrefmark{1}}\\ \vspace{-2mm}
			\IEEEauthorblockA{\IEEEauthorrefmark{1}Friedrich-Alexander-Universit\"at Erlangen-N\"urnberg (FAU), Germany\\
				\IEEEauthorrefmark{2}German Aerospace Center (DLR),  Microwaves and Radar Institute, Weßling, Germany\\
				\vspace{-8mm}}}
			\fi
	\maketitle
	\begin{abstract} 
The integration of unmanned aerial vehicles (UAVs) with radar imaging sensors has revolutionized the monitoring of dynamic and local Earth surface processes by enabling high-resolution and cost-effective remote sensing. This paper investigates the optimization of the sensing accuracy of a UAV swarm deployed to perform multi-baseline interferometric synthetic aperture radar (InSAR) sensing. In conventional single-baseline InSAR systems, two synthetic aperture radar (SAR) images of the same area acquired from two different angles are used to generate a digital elevation model (DEM) of the target area. Multi-baseline InSAR extends this concept by aggregating multiple acquisitions from different angles, thus significantly enhancing the vertical accuracy of the DEM.  The heavy computations required for this process are performed on the ground and, therefore, the radar data are transmitted in real time to a ground station (GS) via a frequency-division multiple access (FDMA) air-to-ground backhaul link. {This is the first work that studies the optimization of a UAV swarm deployed for multi-baseline InSAR applications. We focus} on {optimizing} the sensing precision by minimizing the height error of the averaged DEM while simultaneously ensuring sensing and communication quality-of-service (QoS) {requirements}. To this end, the UAV formation, velocity, and communication power are jointly optimized using evolutionary algorithms (EAs). Our approach is benchmarked against established optimization methods, including genetic algorithms (GAs), simulated annealing (SA), and deep reinforcement learning (DRL) techniques. Numerical results show that the proposed {method} outperforms these baseline schemes and achieves sub-decimeter vertical accuracy in several scenarios. These findings {underscore} the potential of coordinated UAV swarms for delivering high-precision and real-time Earth observations through radar interferometry.
	\end{abstract}
	\section{Introduction}\label{sec:introduction}
	\Ac{sar} is a well-established remote sensing technique that generates high-resolution \ac{2d} images by processing the scattered microwave signals reflected from a target area \cite{radar_book}. In \ac{sar} systems, the radar antenna moves along a predefined trajectory to enhance the spatial resolution of the final image through coherent processing \cite{radar_book}. While \ac{sar} has traditionally been employed on spaceborne and airborne platforms, recent advancements have enabled its integration into \acp{uav} for short-range and on-demand applications. Drone-based \ac{sar} systems are deemed a promising solution for monitoring dynamic environments and enabling real-time high-resolution \ac{sar} imaging \cite{experimental1,experimental2,insar_exp1,insar_exp2,insar_exp3}. In this context, various measurement campaigns using \ac{uav}-\ac{sar} networks have been conducted {for} different applications, such as climate monitoring~\cite{climate_change}, buried object detection~\cite{experimental1, experimental2}, and the generation of \acp{dem}~\cite{insar_exp1,insar_exp3}. {Compared to spaceborne and airborne carriers, equipping \ac{sar} onboard lightweight drones offers greater flexibility and new degrees of freedom, therefore, several studies have targeted the optimization of \ac{uav}-\ac{sar} systems.} In \cite{optimization1}, the authors formulated the resource allocation problem for \ac{uav}-\ac{sar} networks as a \ac{moop}, where the bistatic angles, radar resolution, and \ac{nesz} were considered in the objective function. In \cite{amine2, amine4}, the authors focused on maximizing the \ac{sar} coverage under multiple sensing and communication constraints. Additionally, the authors in \cite{sun2} studied the optimization of geosynchronous \ac{sar} systems and formulated the problem as a \ac{moop} to optimize the trajectory and resolution of the system.
	\subsection{Single-baseline \ac{insar}}
	\Ac{insar} is an advanced \ac{sar} technology that exploits the phase difference between two or more \ac{sar} acquisitions to estimate surface elevation and/or ground displacement \cite{cramer}. These acquisitions can be obtained using either a single platform at different times (temporal baseline) or multiple spatially separated antennas at the same time (spatial or across-track baseline, where the baseline represents the distance between the used sensors). In particular, single-baseline across-track \ac{insar} refers to the configuration where two antennas acquire data simultaneously from different spatial positions, resulting in a purely spatial baseline with no temporal decorrelation \cite{insar1}. Here, the baseline induces a phase difference between the received signals, which can be processed to generate an interferogram that contains for each pixel its corresponding phase difference. This interferogram {includes} information about the terrain elevation, enabling the generation of high-resolution \acp{dem} of the imaged scene \cite{added_krieger}. In this context, important performance metrics include the coherence, which represents the cross-correlation between the two compared \ac{sar} images, and the height error, which represents the vertical precision of the \ac{dem} \cite{coherence1}. 
	Across-track \ac{insar} has been extensively studied for airborne and spaceborne systems \cite{cramer}. However, its adaptation to \ac{uav}-based platforms remains relatively underexplored. While several experimental studies have demonstrated the feasibility of \ac{uav}-based \ac{insar} systems~\cite{insar_exp1,insar_exp2,insar_exp3}, there is a notable lack of systematic research on their performance optimization. Technical challenges such as platform trajectory {design}, payload capacity {constraints}, and the need for effective resource allocation schemes pose significant hurdles. Addressing these challenges is critical {for} advancing low-altitude drone-based \ac{insar} systems. In \cite{optimization2}, the authors studied repeat-pass \ac{uav}-based \ac{3d} \ac{insar} sensing, where metrics such as cross-range resolution, peak sidelobe ratio, integrated sidelobe ratio, and flight time were considered. Furthermore, the authors {of} \cite{optimization1} studied the resource allocation and optimization of multi-\ac{uav} \ac{sar} systems, where the spatial configuration, transmit power, and frequency band allocation were considered. In \cite{amine3,amine4}, the authors studied \ac{insar} coverage maximization under communication and sensing requirements for single-baseline \ac{uav}-based \ac{insar}.
	\subsection{Multi-baseline \ac{insar}} Single-baseline \ac{insar} systems face a fundamental trade-off between phase unwrapping reliability and vertical resolution. Here, phase unwrapping refers to the process of resolving the inherent $2\pi$ ambiguity in the measured phase to obtain absolute phase values, whereas the vertical resolution represents the precision with which elevation differences can be {estimated} in the generated \ac{dem} \cite{cramer}. On {the} one hand, short baselines provide high coherence and more reliable phase estimates, facilitating the phase unwrapping process. On the other hand, short baselines offer limited sensitivity to elevation, resulting in lower vertical resolution in the final \ac{dem}. Conversely, longer baselines improve the vertical resolution but suffer from increased decorrelation, making the phase unwrapping more challenging \cite{insar1}. Multi-baseline \ac{insar} addresses {these limitations} by combining multiple acquisition platforms or repeated observations to form several interferometric pairs with varying baseline lengths and orientations \cite{added_krieger2,added_krieger3}. This approach enables the generation of a diverse set of interferograms, allowing {\ac{insar} processing} to benefit from both short and long baselines. Compared to single-baseline \ac{insar}, the optimization of multi-baseline \ac{insar} is more challenging, primarily due to the use of multiple sensors. For instance, with $I>2$ sensors in operation, the system's degrees of freedom increase significantly, introducing intricate interdependencies. In fact, unlike single-baseline \ac{insar}, where only one baseline and one \ac{dem} are influenced by the sensor geometry, in a multi-baseline setup, adjusting the position of a single sensor simultaneously affects $I-1$ baselines and consequently, $I-1$ \acp{dem}. This strong mutual coupling between sensors, along with the high dimensionality of the problem, increases the overall computational complexity required for system optimization. Furthermore, the large volume of data generated by multiple sensors places additional strain on data transmission and processing. Although multi-baseline \ac{insar} is a well-established technique, especially in spaceborne systems \cite{added_krieger3}, it remains underexplored in the context of \acp{uav}. Additionally, the system geometry, defined by the spatial distribution of the antennas, strongly affects the primary peformance metric of the generated \ac{dem}, which is the height accuracy \cite{fusion_sigma_h}. Therefore, minimizing the \ac{dem} height error is an important yet challenging optimization problem that has not been explicitly formulated or addressed in prior work.
	\subsection{Main Contributions}
In this paper, we consider a \ac{uav} swarm equipped with \ac{sar} systems, designed to perform multi-baseline \ac{insar} sensing. Each \ac{uav} pair forms a single-baseline \ac{insar} system capable of generating a \ac{dem}. All \acp{dem} are then combined into a final high-precision \ac{dem} using averaging techniques. Our goal is to minimize the height error of the final \ac{dem} by jointly optimizing the \ac{uav} formation, velocity, and communication power allocation, subject to sensing and communication constraints. {To the best of our knowledge, this is the first paper to address the optimization of sensing accuracy in multi-\ac{uav} multi-baseline \ac{insar} systems.} Our main contributions can be summarized as follows:

\begin{itemize}
	\item We formulate the height error minimization problem for a multi-baseline InSAR system consisting of \( I \) \acp{uav} as a non-convex optimization problem, accounting for {practical} system constraints related to coverage, phase unwrapping performance, data offloading, and collision avoidance.
	
	\item To address this challenging problem, we propose a novel co-evolutionary algorithm in which two distinct species collaboratively explore the optimization search space. The original problem is decomposed into two simpler sub-problems, {where} each species {is} tasked with exploring the search space corresponding to its respective sub-problem. Non-parameterized \ac{pso} is then used as the evolutionary optimizer to simultaneously evolve both species.
	
	\item We compare the proposed algorithm with several established optimization methods, such as \acp{ga}, \ac{sa}, and \ac{drl}, and demonstrate its superiority across various system parameters.
	
	\item We analyze the potential of employing \ac{uav} networks to enhance the sensing precision and explore the critical interplay between sensing and communication requirements. {Furthermore, we show that optimizing the \ac{uav} formation results in height errors below 10 cm and identify the optimal number of sensors required for a given sensing objective.}
\end{itemize}
We note that this article extends the corresponding conference version~\cite{amine5}. In~\cite{amine5}, the analysis was limited to dual-baseline \ac{uav}-based systems, considering only two {pairs} of \ac{sar} acquisitions. Furthermore, instead of minimizing the actual height error, an upper bound was derived in \cite{amine5} and employed to simplify the analysis. Additionally, in~\cite{amine5}, the \ac{uav} swarm velocity was assumed to be fixed and, therefore, was not optimized.

The remainder of this paper is organized as follows. In Section \ref{sec:system_model}, we present the system model. In Section \ref{sec:problem_formulation}, the optimization problem for the minimization of the DEM height error is formulated. The proposed solution is presented in Section \ref{sec:solution}. In Section \ref{sec:simulation_results}, the performance of the proposed method is {evaluated} based on numerical simulations. Finally, conclusions are drawn in Section \ref{sec:conclusion}.

{\em Notations}:
In this paper, lower-case letters $x$ refer to scalar {variables}, while boldface lower-case letters $\mathbf{x}$ denote vectors.  $\{a, ..., b\}$ denotes the set of all integers between $a$ and $b$.  $|\cdot|$ denotes the absolute value operator, while $[\cdot]^+$ denotes $\max(0,\cdot)$. $\mathbb{R}^{N}$ represents the set of all $N$-dimensional vectors with real-valued entries. For a vector {$\mathbf{x}=(x_1,...,x_N)^T\in\mathbb{R}^{N}$}, $||\mathbf{x}||_2$ denotes the Euclidean norm, whereas  $\mathbf{x}^T$ stands for the  transpose of $\mathbf{x}$. 
For a given set \( \mathcal{S} \), \( \text{card}(\mathcal{S}) \) refers to the number of elements in the set.  For two vectors \( \mathbf{a} \in \mathbb{R}^{N} \) and \( \mathbf{b} \in \mathbb{R}^{N} \) such that $a[n] \leq b[n], \forall n \in \{1, \ldots, N\}$, \( \mathcal{U}(\mathbf{a}, \mathbf{b}) \) refers to the uniform distribution over the range defined by \( \mathbf{a} \) and \( \mathbf{b} \), i.e., a uniform distribution between the components of the two vectors. For an integer \( N \), \( \mathbf{0}_N \) and \( \mathbf{1}_N \) refer to the all-zero and all-one vectors of size \( N \), respectively.
	\section{System Model} \label{sec:system_model}
	
	\begin{figure}
		\centering
		\ifonecolumn
		\includegraphics[width=4in]{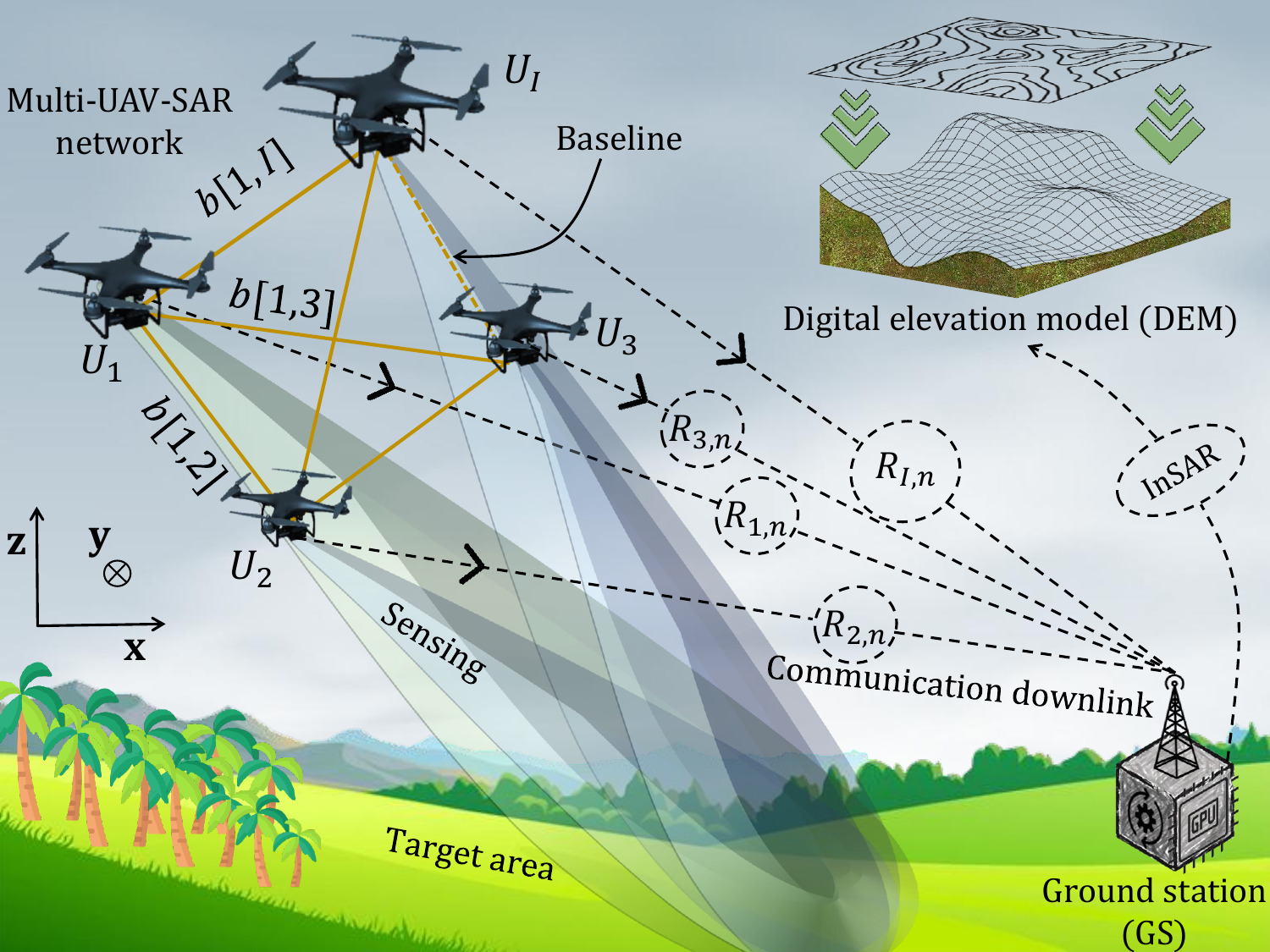}
		\else
		\includegraphics[width=0.9\columnwidth]{figures/SystemModel.pdf}
		\fi
		\caption{ Multi-baseline \ac{insar} sensing system comprising $I$ \ac{uav}-\ac{sar} systems {as well as} a \ac{gs} for real-time data offloading. The \ac{uav} swarm is moving perpendicular to the plane of the figure. }
		\label{fig:system-model}
	\end{figure} 
	We consider $I \in \mathbb{N}, I\geq 3$ rotary-wing \acp{uav}, {denoted} by $U_i, i \in \mathcal{I}=\{1,...,I\}$, performing \ac{insar} sensing over a target area. The \ac{uav} swarm operates in standard \ac{insar} mode \cite{cramer}, where $U_1$, the master drone, transmits and receives radar echoes, while $U_i, i \in\mathcal{I}\setminus\{1\}$, the slave drones, only receive. This is equivalent to setting $P_{\mathrm{rad},i}=0, \forall i \in \mathcal{I}\setminus\{1\}$, where $P_{\mathrm{rad},i}, \forall i \in \mathcal{I},$ is the radar transmit power of \ac{uav} $U_i$. We use a \ac{3d} coordinate system, where the  $x$-,  $y$-, and $z$-{axes} represent the cross-range {direction}, the azimuth {direction}, and the altitude, respectively. The mission time $T$ is divided into $N$ slots {of duration $\delta_t$}, with $T = N\cdot\delta_t$, and set $\mathcal{N}=\{1,...,N\}$ denotes the set of all time slots. The drone swarm {forms} a multi-baseline interferometer with $\frac{I(I-1)}{2}$  observations acquired by all existing \ac{uav} pairs, see Figure \ref{fig:system-model}.  The considered \ac{uav}-\ac{sar} systems operate in stripmap mode \cite{book1} and fly at a constant but {adjustable} velocity, $v_y$, following a linear trajectory  that is parallel to a line, denoted by $l_t$, {which is} parallel to {the} $y$-axis and {passes} in time slot $n \in \mathcal{N}$ through reference point $\mathbf{p}_t[n]=(x_t,y[n],0)^T \in \mathbb{R}^3$. In other words, line $l_t$ is the swath center line. The position of $U_i$ {in} time slot $n \in \mathcal{N}$ is $\mathbf{q}_i[n]=(x_i,y[n],z_i)^T$, with the $y$-axis position vector $\mathbf{y}=(y[1]=0,y[2], ..., y[N])^T\in\mathbb{R}^{N}$ given by: 
	\begin{align}
		y[n+1]=y[n]+v_y\delta_t, \forall n \in \mathcal{N}\setminus\{N\}.
	\end{align}
	For simplicity, we denote the position of $U_i$ in the across-track plane (i.e., $xz-$plane) by $\mathbf{q}_i=(x_i,z_i)^T \in \mathbb{R}^2, \forall i \in \mathcal{I}$. In \ac{sar} interferometry, the distance between two  sensors $U_i$ and $U_j, i\neq j,$ is a key metric and is referred to as the interferometric baseline. The {considered} system can perform multi-baseline \ac{insar} due to the existence of $\frac{I(I-1)}{2}$ \ac{uav} pairs, i.e., baselines.  The interferometric baseline for \ac{uav} pair $(U_i,U_j), i < j, i,j \in \mathcal{I},$ is denoted by $b[i,j]$ and given by:
	\begin{align}
		b[i,j] = ||\mathbf{q}_i -\mathbf{q}_j||_2, \forall\, i, j \in \mathcal{I},\; i < j.
	\end{align}
	\begin{comment}
		\begin{figure}
			\centering
			\ifonecolumn
			\includegraphics[width=4in]{figures/SystemModel2.pdf}
			\else 
			\includegraphics[width=0.9\columnwidth]{figures/SystemModel2.pdf}
			\fi
			\caption{Illustration of the multi-static \ac{uav} formation with $K$ \ac{insar} pairs  in the across-track plane while focusing on the $k^{\rm th}$ pair $(U_0,U_k)$.}
			\label{fig:system-model-2}
		\end{figure}
	\end{comment}
	The perpendicular baseline, {denoted by $b_{\bot}$,} is the magnitude of the projection of $(U_i,U_j)$'s baseline vector perpendicular to $U_i$'s \ac{los} to $\mathbf{p}_t[n]$ and is given by:
	\begin{equation}\label{eq:perpendicular_baseline}
		b_{\bot}[i,j]=
		b[i,j] \cos\Big({\theta_i}- \alpha[i,j]\Big), \forall\, i, j \in \mathcal{I},\; i < j,
	\end{equation} 
	where $\theta_i$ is $U_i$'s look angle  and is adjusted {such that} the beam footprint is centered around $\mathbf{p}_t$ (i.e., the radar swath is centered \ac{wrt} $l_t$):
	\begin{equation}
		\theta_i= \arctan\left(\frac{x_t-x_i}{z_i}\right), \forall i \in \mathcal{I}.
	\end{equation} 
	Here, $\alpha[i,j]$ is the angle between the {baseline vector of $(U_i,U_j)$} and the horizontal plane and is given by: 
	\begin{equation}
		\alpha[i,j]=\arctan\left(\frac{z_j-z_i}{x_j-x_i}\right), \forall\, i, j \in \mathcal{I},\; i < j.
	\end{equation}
	\subsection{\ac{insar} Performance}
	Next, we introduce the relevant \ac{insar} sensing performance metrics.
	\subsubsection{\ac{insar} Coverage}
	Let $r_i, i \in \mathcal{I}$, {denote} $U_i$'s slant range \ac{wrt} $ \mathbf{p}_t[n]$. {The} slant range is independent of time and is given by:
	\begin{align}
		r_i= \sqrt{ (x_i -x_t)^2 +  z_i^2  }, \forall i \in  \mathcal{I}.
	\end{align} The swath width of $U_i$ can be approximated as follows \cite{book1}:
	\begin{align}\label{eq:swath_width}
		S_i= \frac{\Theta_{\rm 3 dB}r_i }{\cos(\theta_i)},\forall i \in \mathcal{I},
	\end{align}
	where   $\Theta_{\mathrm{3dB}}$ is the -3 dB beamwidth in elevation. The total \ac{insar} coverage can be approximated as follows: 
	\begin{equation}
		C_{\mathrm{tot}}=N \min\limits_{i \in \mathcal{I}} \left(S_i\right) v_y \delta_t.
	\end{equation}
	\subsubsection{\Ac{insar} Coherence} A key performance metric for \ac{insar} is the coherence, which represents the normalized cross-correlation between a pair of \ac{sar} images. For the \ac{insar} pair formed by the images acquired by the pair $(U_i, U_j)$, the total coherence is denoted by $0 \leq \gamma[i,j] \leq 1$. A coherence value close to 1 indicates high similarity between the two images, which is essential for accurate phase measurements and reliable interferometric analysis. The total coherence can be decomposed into several decorrelation sources as follows:
	\begin{equation}
		\gamma[i,j]= \gamma_{\mathrm{Rg}}[i,j] \gamma_{\mathrm{SNR}}[i,j]\gamma_{\rm other}[i,j], \forall\, i, j \in \mathcal{I},\; i < j.
	\end{equation} 
	Here, for a given interferometric pair $(U_i,U_j)$,  $\gamma_{\mathrm{Rg}}[i,j]\in [0, 1],$ is the baseline  decorrelation, $\gamma_{\mathrm{SNR}}[i,j]\in [0, 1]$ is the \ac{snr} decorrelation, and $\gamma_{\rm other}[i,j]\in [0, 1]$ represents the contributions from all other decorrelation sources, which are independent of the \ac{uav} formation. The \ac{snr} decorrelation of pair $(U_i,U_j)$ is affected by {the} \acp{snr} {of both \acp{uav}} and is given by  \cite{snr_equation}:
{\small	\begin{align} \label{eq:snr_decorrelation}
	\gamma_{\mathrm{SNR}}[i,j]= \frac{1}{\sqrt{\left(1+\mathrm{SNR}^{-1}_{i}\right)}}\frac{1}{\sqrt{\left(1+\mathrm{SNR}^{-1}_{j}\right)}}, \forall\, i, j \in \mathcal{I},\; i < j,
	\end{align}}
	where $\mathrm{SNR}_1$ denotes the \ac{snr} of the mono-static acquisition received by $U_1$ given by \cite{snr_equation}:
	\begin{equation}\label{eq:monostatic_snr}
		\mathrm{SNR}_{1}=\frac{\sigma_0 P_{\mathrm{rad},1}\; G_t\; G_r \lambda^3 c \tau_p \mathrm{PRF} }{4^4 \pi^3  v_y \sin(\theta_1) r_1^3 k_b T_{\mathrm{sys}}  \; B_{\mathrm{Rg}} \; F \; L}.
	\end{equation}
	 Here, $\sigma_0$ is the normalized backscatter coefficient, $P_{\mathrm{rad},1}$, is the radar transmit power of \ac{uav} $U_1$, $G_t$ and $G_r$ are the transmit and receive antenna gains, respectively, $\lambda$ is the radar wavelength, $c$ is the speed of light, $\tau_p$ is the pulse duration, $\mathrm{PRF}$ is the pulse repetition frequency, $k_b$ is the Boltzmann constant, $T_{\mathrm{sys}}$ is the receiver noise reference temperature, $B_{\mathrm{Rg}}$ is the bandwidth of the radar pulse, $F$ is the noise figure of the receiver, and  $L$ represents the total radar losses. Assuming small bi-static angles $|\theta_1-\theta_i|, \forall i$, which is typically the case for \ac{insar} applications \cite{cramer}, the bi-static \ac{snr} for $U_i, i\geq 2$, denoted by $\mathrm{SNR}_i$, {can be} approximated as follows \cite{amine4}:
	\begin{equation}\label{eq:bistatic_snr}
		\mathrm{SNR}_{i}\approx \frac{\sigma_0 P_{\mathrm{rad},1}\; G_t\; G_r \lambda^3 c \tau_p \mathrm{PRF} }{4^4 \pi^3  v_y \sin(\theta_i) r_1^2 r_i k_b T_{\mathrm{sys}}  \; B_{\mathrm{Rg}} \; F \; L}, \forall i \in \mathcal{I}\setminus\{1\}.
	\end{equation}
	Furthermore, the baseline decorrelation reflects the loss of coherence caused by the different angles used for the acquisition of both \ac{insar} images, {and can be expressed as follows} \cite{victor2}: 
	\begin{equation} \gamma_{\mathrm{Rg}}[i,j]=\frac{1}{B_p} \left[ \frac{2+B_p}{1+\mathcal{X}[i,j]} -\frac{2-B_p}{1+\mathcal{X}^{-1}[i,j]} \right],\label{eq:baseline_decorrelation}
	\end{equation}
	where $B_{ p}=\frac{B_{\mathrm{Rg}}}{f_0}$ is the fractional bandwidth, $f_0$ is the radar center frequency, and function $\mathcal{X}[i,j]$ is {given} by \cite{victor2}:
	\ifonecolumn 
	\begin{equation}
		\mathcal{X}[i,j]=\frac{ \sin(\max(\theta_i,\theta_j))}{ \frac{1}{2} \left(\sin(\theta_i)+\sin(\theta_j)\right)}, \forall\, i, j \in \mathcal{I},\; i < j. \label{eq:baseline_decorrelation_X}
	\end{equation} 
	\else
	\begin{align}
		\mathcal{X}[i,j]=\frac{ \sin(\max(\theta_i,\theta_j))}{ \frac{1}{2} \left(\sin(\theta_i)+\sin(\theta_j)\right)}, \forall\, i, j \in \mathcal{I},\; i < j. \label{eq:baseline_decorrelation_X}
	\end{align}
	\fi
	
	\subsubsection{Height of Ambiguity (HoA)} 
	To extract the height of ground targets, \ac{insar} relies on the phase difference between images acquired by  $(U_i, U_j), i < j$. Assuming that sensors $U_i$ and $U_j$ measure phases $\phi_i$ and $\phi_j$, respectively, the interferometric phase is given by $\Phi_{i,j} = \phi_i - \phi_j$. The elevation of a target above ground introduces a phase shift in the measured phase $\Phi_{i,j}$, making it possible to construct a \ac{dem} of the observed area. However, since the phase $\Phi_{i,j}$ is $2\pi$-cyclic, the corresponding height information is inherently ambiguous. The height difference corresponding to a complete $2\pi$ phase cycle is known as the \ac{hoa}. The \ac{hoa} is a key  parameter, as it determines the system’s sensitivity to topographic height variations \cite{snr_equation}. The \ac{hoa} for sensor pair $(U_i, U_j)$ is given by \cite{snr_equation}:
	\begin{align}\label{eq:height_of_ambiguity}
		h_{\mathrm{amb}}[i,j]=\frac{\lambda r_i \sin(\theta_i)}{ b_{\perp}[i,j]}, \forall\, i, j \in \mathcal{I},\; i < j.
	\end{align}
	\begin{comment}Note that a large \ac{uav} baseline leads to a small \ac{hoa} value, which results in increased errors for interferometric phase unwrapping \cite{coherence1}. However, a very small \ac{uav} baseline, i.e., a large \ac{hoa}, is not a good choice either, as we explain in the following.\end{comment}
	\subsubsection{{\ac{dem}} Height Accuracy}
	The height error of {the \ac{dem} referes to the precision with which the elevation of ground targets {can be} estimated. The height error  {acquired by} \ac{insar}} pair $(U_i,U_j)$ is given by \cite{snr_equation}:
	\begin{equation} \label{eq:height_error}
		\sigma_{h}[i,j] =h_{\mathrm{amb}}[i,j] \frac{\sigma_{\Phi}[i,j] }{2\pi}, \forall\, i, j \in \mathcal{I},\; i < j,
	\end{equation}
	where {$\sigma_{\Phi}$} is the random error in the interferometric phase, whose standard deviation can be approximated in the case of high interferometric coherences by the Cramér–Rao bound  \cite{cramer}: 
	\begin{equation}\label{eq:phase_error}
		\sigma_{\Phi}[i,j]=\frac{1}{\gamma[i,j]} \sqrt{ \frac{1-\gamma^2[i,j]}{2n_L}}, \forall\, i, j \in \mathcal{I},\; i < j,
	\end{equation}
	{where $n_L$ is the number of independent looks {employed, i.e.,} {$n_L$} adjacent pixels of the interferogram {are averaged} to improve the phase estimation \cite{cramer}.}\par
	 {To enhance accuracy,} fusion of the $\frac{I(I-1)}{2}$ \ac{insar} {\acp{dem}} is performed based on inverse-variance weighting, such that the {height of} an arbitrary target {estimated} by \ac{insar} pair $(U_i,U_j), i, j \in \mathcal{I},\; i < j$, denoted by $h[i,j]$, {is} weighted by $\beta[i,j]=\frac{1}{	\sigma^2_{h}[i,j]}$, and {averaged} as $\frac{\sum\limits_{i<j \in \mathcal{I}}h[i,j] \beta[i,j]}{\sum\limits_{i<j \in \mathcal{I}}\beta[i,j]}$. The final height error of the fused \ac{dem} is {given} by \cite{fusion_sigma_h}: 
	\begin{align}
		\sigma_{h}= \sqrt{\frac{\sum\limits_{ i<j \in \mathcal{I}}\beta^2[i,j]\sigma^2_{h}[i,j]}{\left(\sum\limits_{i<j \in \mathcal{I}}\beta[i,j]\right)^2}}. \label{eq:final_height_error}
	\end{align}
	\subsection{Communication Performance}
	We consider real-time offloading of the radar data to a \ac{gs}, where all \acp{uav} employ \ac{fdma}. The  instantaneous   communication transmit power  consumed by \ac{uav} $U_i$ is given by $\mathbf{P}_{\mathrm{com},i}=(P_{\mathrm{com},i}[1],...,P_{\mathrm{com},i}[N])^T \in \mathbb{R}^N, i \in \mathcal{I}$.
	We denote the location of the \ac{gs} by $\mathbf{g}= (g_x, g_y, g_z)^T \in \mathbb{R}^3$ and the distance from $U_i$ to the \ac{gs} by    $d_{i,n} = ||\mathbf{q}_i[n]-\mathbf{g} ||_2, \forall i \in \mathcal{I}, \forall n \in \mathcal{N}.$ Thus, {adopting} the free-space path loss model and \ac{fdma}, the
	instantaneous throughput  from $U_i, \forall i \in \mathcal{I},$ to the \ac{gs} is given by:
	\begin{align}
		&R_{i,n}= B_{c,i} \; \log_2\left(1+\frac{P_{\mathrm{com},i}[n] \;\beta_{c,i}}{d_{i,n}^2}\right), \forall i \in \mathcal{I}, \forall n \in \mathcal{N},
	\end{align}
	where  {$B_{c,i}$ is $U_i$'s fixed} communication bandwidth and $\beta_{c,i}$ is {the} reference channel gain\footnote{The reference channel gain is the channel power gain at a reference distance of 1 m.} divided by the noise variance. The sensing data rate generated by {$U_i$'s \ac{sar} system} can be approximated as follows \cite{amine4}: 
	\ifonecolumn
	\begin{equation}
		R_{\mathrm{min},i}=\frac{n_B B_{\mathrm{Rg}} \mathrm{PRF}}{c}\Big(c \tau_p+\frac{z_i}{\cos(\theta_i+\frac{\Theta_{\mathrm{3dB}}}{2})}- \frac{z_i}{\cos(\theta_i-\frac{\Theta_{\mathrm{3dB}}}{2})}\Big), \forall i \in \mathcal{I},
	\end{equation}
	\else
	\begin{align}
		&R_{\mathrm{min},i}=\frac{n_B B_{\mathrm{Rg}} \mathrm{PRF}}{c}\Big(c \tau_p+\frac{z_i}{\cos(\theta_i+\frac{\Theta_{\mathrm{3dB}}}{2})}-\\ \notag &\frac{z_i}{\cos(\theta_i-\frac{\Theta_{\mathrm{3dB}}}{2})}\Big), \forall i \in \mathcal{I},
	\end{align}
	\fi
	where $n_B$ is the number of bits per complex sample. 
	\subsection{Energy Consumption} 
	The propulsion power consumed by each drone of the \ac{uav} swarm, flying with steady speed $v_y$, is given by \cite{propulsion}: 
	\ifonecolumn
	\begin{equation}\label{eq:propulsion_power}
		P_{\mathrm{prop}}=  P_0 \left(1+\frac{3v_y^2}{U^2_{\mathrm{tip}}}\right)+P_I\left( \sqrt{1+\frac{v^4_y}{4v_0^4}}-\frac{v^2_y}{2v_0^2}\right)^{\frac{1}{2}}+\frac{1}{2}d_0 \rho s A_e v_y^3.
	\end{equation}
	\else 
	\begin{gather}\label{eq:propulsion_power}
		P_{\mathrm{prop}}=  P_0 \left(1+\frac{3v_y^2}{U^2_{\mathrm{tip}}}\right)+P_I\left( \sqrt{1+\frac{v^4_y}{4v_0^4}}-\frac{v^2_y}{2v_0^2}\right)^{\frac{1}{2}}+\notag\\ \frac{1}{2}d_0 \rho s A_e v_y^3.
	\end{gather}
	\fi
	Here, $P_0=\frac{\delta_u}{8}\rho s A_e \Omega^3 R^3$ and  $P_I=(1+k_u) \frac{W_u^{\frac{3}{2}}}{\sqrt{2\rho A_e}}$ are two constants, $v_0=\sqrt{\frac{W_u}{2 \rho A_e}}$ is the mean rotor induced velocity {while hovering}, $U_{\mathrm{tip}}$ is the tip speed of the rotor blade, $d_0$ is the fuselage drag ratio,  $\delta_u$ is the profile drag coefficient, $\rho$ is the air density, $s$ is the rotor solidity, $A_e$ is the rotor disc area, $\Omega$ is the blade angular velocity, $R$ is the rotor radius, $k_u$ is a  correction factor, and $W_u$ is the aircraft weight in Newton. The total energy consumed by \ac{uav} $U_i$ can be expressed as follows:
	\begin{equation}
		E_i= \underbrace{T P_{\mathrm{prop}}}_{\text{Operation}}+\underbrace{T P_{\mathrm{rad},i}}_{\text{Sensing}}+ \underbrace{\sum\limits_{n=0}^{N-1}\delta_t P_{\mathrm{com},i}[n]}_{\text{Offloading}}, \forall i \in \mathcal{I}.
	\end{equation}
	\section{Problem Formulation} \label{sec:problem_formulation}
	In this paper, we aim to minimize the height error of the final \ac{dem} $\sigma_h$ by jointly optimizing the \ac{uav} formation $\mathcal{Q}=\{\mathbf{q}_i,\forall i \in \mathcal{I}\}$, the \ac{uav} swarm velocity $v_y$, and the instantaneous communication transmit powers  $\mathcal{P}=\{\mathbf{P}_{\mathrm{com},i}, \forall i \in \mathcal{I}\}$, while satisfying quality-of-service constraints related to sensing and communications. To this end, we formulate the following optimization problem: 
	\begin{alignat*}{2} 
		&(\mathrm{P}):\min_{\mathcal{Q},\mathcal{P}, v_y} \hspace{3mm}  \sigma_{h}  & \qquad&  \\
		\text{s.t.} \hspace{3mm} &\mathrm{C1}: z_{\mathrm{min}} \leq z_i \leq z_{\mathrm{ max}}, \forall i  \in \mathcal{I},               &      &  \\ &\mathrm{C2}: \theta_{\mathrm{min}} \leq \theta_i \leq \theta_{\mathrm{max}}, \forall i \in \mathcal{I},            &      &  \\ & \mathrm{C3}: v_{\mathrm{min}} \leq v_y \leq v_{\mathrm{max}},          &      &  \\
		& \mathrm{C4}: 0 \leq P_{\mathrm{com},i}[n]  \leq P_{\mathrm{com}}^{\mathrm{max}}, \forall i \in \mathcal{I}, \forall n \in \mathcal{N},    & &\\
		&  \mathrm{C5}: b[i,j] \geq d_{\mathrm{min}},\forall i < j  \in \mathcal{I},  &      &     
		\\
		&  \mathrm{C6}: C_{\mathrm{tot}}\geq C_{\rm min},    &      &     
		\\
		&    \mathrm{C7}: h_{\mathrm{amb}}[i,j]\geq  h_{\mathrm{amb}}^{\mathrm{min}}, \forall (i,j) \in \mathcal{I}_{\mathrm{phase}},             &      & \\
		& \mathrm{C8}: R_{i,n} \geq R_{\mathrm{min},i}, \forall i \in \mathcal{I}, \forall n \in \mathcal{N},       & &  \\
		& \mathrm{C9}: E_i \leq E_{\rm max}, \forall i \in \mathcal{I}.   & &  
	\end{alignat*}
	Constraint $\mathrm{C1}$ specifies the minimum and maximum allowed flying altitude, denoted by $z_{\mathrm{min}}$ and $z_{\mathrm{max}}$,  respectively. Constraint $\mathrm{C2}$ specifies the minimum and maximum slave look angle, denoted by $\theta_{\rm min}$ and $\theta_{\rm max}$, respectively. Similarly, constraint $\mathrm{C3}$ defines the minimum and maximum allowed velocity for the \ac{uav} swarm, denoted by $v_{\mathrm{min}}$ and $v_{\mathrm{max}}$, respectively. Constraint $\mathrm{C4}$ imposes a maximum communication transmit power,  $P_{\mathrm{com}}^{\mathrm{max}}$. Constraint $\mathrm{C5}$ enforces a minimum safety distance $d_{\mathrm{min}}$  between {any two} \acp{uav}. Constraint $\mathrm{C6}$ imposes a minimum \ac{sar} coverage $C_{\mathrm{min}}$.  Constraint $\mathrm{C7}$ imposes a minimum \ac{hoa}, $h_{\mathrm{amb}}^{\mathrm{min}}$, {required for} phase unwrapping \cite{coherence1}. {Here, we assume that $M$ baselines are required to satisfy the phase unwrapping constraint. Consequently, this constraint is enforced only for a subset of drone pairs, denoted by $\mathcal{I}_{\mathrm{phase}} = \{(i,j) \mid i, j \in \mathcal{I},\, i < j\}$, where $\mathrm{card}(\mathcal{I}_{\mathrm{phase}}) = M$ and $M \leq \frac{I(I-1)}{2}$. Note that the special case $M = \frac{I(I-1)}{2}$ corresponds to enforcing the phase unwrapping constraint for all baselines.} Constraint $\mathrm{C8}$ ensures {the} minimum {required} data rate for sensor $U_i$, denoted by $R_{\mathrm{min},i}, \forall i \in \mathcal{I}$. Constraint $\mathrm{C9}$ limits the  total energy consumed by \ac{uav} $U_i$ to  $E_{\rm max}, \forall i \in \mathcal{I}$.\par 
	Problem $\mathrm{(P)}$ is a highly challenging non-convex optimization problem. First, the {expression for} objective function $\sigma_h$ {is} complex involving multiple trigonometric functions and coupling between the \ac{uav} positions $\mathcal{Q}$ and the velocity of the \ac{uav} swarm $v_y$. Second, constraints $\mathrm{C2}$, $\mathrm{C6}$, and $\mathrm{C7}$ are non-convex due to the trigonometric functions {involved}. Third, constraint  $\mathrm{C8}$ is non-convex as the minimum data rate, $R_{\mathrm{min},i}$, varies with the position of drone $U_i, \forall i$. Additionally, constraint $\mathrm{C5}$ imposes a lower bound on a Euclidean distance and constraint $\mathrm{C9}$ {involves} a non-convex \ac{uav} propulsion power model, thereby making both constraints non-convex. Due to the aforementioned reasons and the problem dimension, given by $\mathrm{card}(\mathcal{Q})+\mathrm{card}(\mathcal{P})+1=N I +2I +1$, problem $\mathrm{(P)}$ is difficult to handle and, therefore, {its solution requires the application of} advanced optimization techniques. 
	
	\section{Solution of the Optimization Problem}\label{sec:solution}
	Given its non-convexity, non-monotonicity, and high dimensionality, traditional non-convex and monotonic optimization methods  {are not suitable for tackling} problem $\mathrm{(P)}$. Therefore, we adopt stochastic optimization techniques, employing a population-based \ac{ea} \cite{evolutionary_algorithms2}.
	\subsection{Co-evolutionary Framework}
	\ifonecolumn
	\begin{algorithm}[t]
		\caption{Parallel PSO-based Co-evolutionary Algorithm}
		\label{alg:coevolution}
		\begin{algorithmic}[1] 
			\Procedure{Co-evolution}{}
			\Statex \textbf{Output:} Solution to problem $\mathrm{(P)}$
			\Statex \hspace{5mm} \textbf{Initialization:}
			\State Initialize iteration number $k_2 \gets 1$
			\State Initialize random population $S_2$ of size $D_2$ and random PSO velocities
			\State Evaluate the initial population $S_2$ using Lines~9--11 of this algorithm.
			\State \label{line:best_global1}Set the local and global best particles and their respective fitness
			\While{$k_2 \leq K_2$}
			\State Update PSO velocities $v_{\mathrm{PSO}, d_2}^{(k_2)}, \forall d_2$, using (\ref{eq:pso_velocity_update_s2}) and (\ref{eq:pso_velocity_wall_s2})
			\State Update population $S_2$; component $l_{d_2}^{(k_2)}[1]$ is updated using (\ref{eq:update_population_s2}) whereas 
			\Statex \hspace{5mm}\hspace{5mm} $\left(l_{d_2}^{(k_2)}[2], \ldots,l_{d_2}^{(k_2)}[NI+1]\right)^T$ are updated based on \textbf{Proposition}~\ref{prop:solution_communication_power}
			\State \textbf{do in parallel for each} $d_2 \in \{1, \ldots, D_2\}$
			\State \hspace{\algorithmicindent} Evolve the associated population $S_{1,d_2}$ by $K_1$ generations using \textbf{Algorithm} \ref{alg:pso} to get  
			\Statex \hspace{\algorithmicindent} \hspace{\algorithmicindent}\hspace{\algorithmicindent}solution: $\mathbf{p}_{\mathrm{best}, d_2}^{(K_1)} \gets \texttt{PSO}(\mathbf{l}_{d_2}^{(k_2)})$
			\State Evaluate the fitness of each $\mathbf{l}_{d_2}^{(k_2)}$ given $\mathbf{p}_{\mathrm{best}, d_2}^{(K_1)}$ using (\ref{eq:fitness_s2})
			\State \label{line:best_global2}Update current local and global best particles in $S_2$ and their respective fitness
			\State Increment iteration number: $k_2 \gets k_2 + 1$
			\EndWhile
			\State \Return best particle of $S_2$ and its associated best particle from species $\{S_{1,d_2}\}_{d_2=1}^{d_2=D_2}$: $\{\mathbf{l}_{\mathrm{best}}^{(K_2)},\mathbf{p}_{\mathrm{best}, d_2}^{(K_1)}\}$.
			\EndProcedure
		\end{algorithmic}
	\end{algorithm}
	\fi
	\begin{figure}[]
		\centering
		\ifonecolumn
		\includegraphics[width=4.5in]{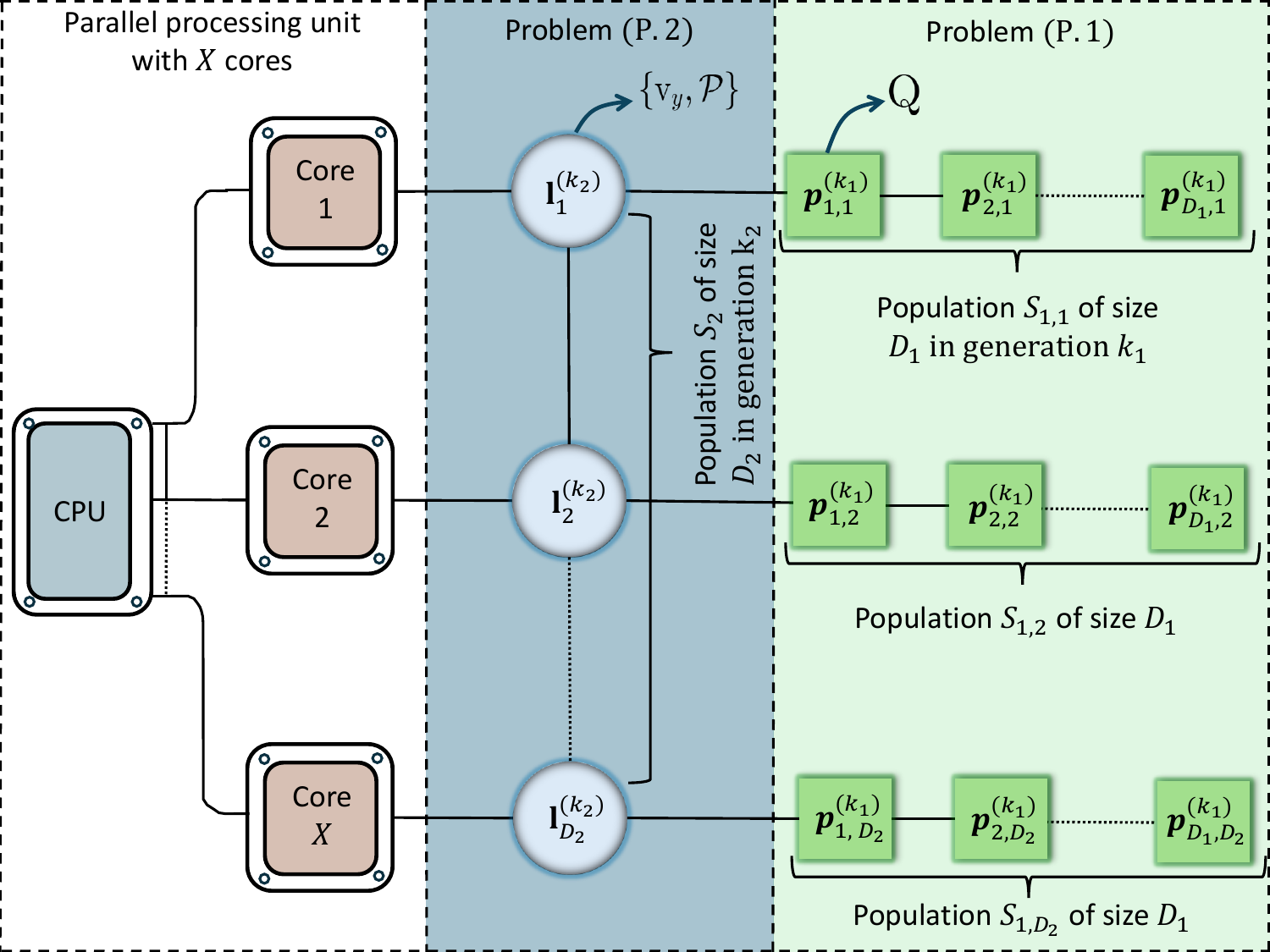}
		\else
		\includegraphics[width=0.9\columnwidth]{figures/diagram.pdf}
		\fi
		\caption{Block diagram of the proposed co-evolutionary algorithm with parallel implementation utilizing a total of $X$ cores. Optimal parallel efficiency is achieved when $X=D_2$, while smaller values of $X<D_2$ are also supported.}
		\label{fig:diagram}
	\end{figure}
	\Acp{ea} have shown great {promise for} solving constrained optimization problems \cite{evolutionary_algorithms2,evolutionary_algorithms}. Yet, \acp{ea} {still often} struggle {in solving} high-dimensional and non-separable problems \cite{coevolution_survey,coevolution_largescale2}. In recent years, co-evolution algorithms, an extension of \acp{ea}, have received increased attention due to their effectiveness in addressing  constrained \cite{coevolution_pso}, large-scale \cite{coevolution_largescale1,coevolution_largescale2}, dynamic \cite{coevolution_dynamic}, and multi-objective \cite{coevolution_multiobjective} optimization problems. In nature, co-evolution refers to the reciprocal evolutionary change between species that interact with each other \cite{coevolution_survey}, such as the co-evolution observed between plants and butterflies \cite{coevolution_butterfly}. In mathematical optimization,  co-evolution consists of decomposing a global problem into smaller sub-problems, each solved by distinct types of populations with different evaluation functions. {Here, a population consists of a set of particles, each representing a candidate solution to the problem, and the evaluation function, also known as the fitness function, is used to assess the quality or performance of candidate solutions based on the optimization goal.} These different types of populations evolve independently based on \ac{ea}-based optimizers while collaboratively contributing to the global solution {for a number of generations, i.e., iterations}. In general, {there is no standardized formulation for} co-evolutionary algorithms, instead, they are defined based on the characteristics of the target problem.\par
			\ifonecolumn
	\begin{algorithm}[t]
		\caption{Particle Swarm Optimization (PSO)}\label{alg:pso}
		\begin{algorithmic}[1]
			\Function{\texttt{PSO}}{$\mathbf{l}_{d_2}^{(k_2)}$}
			\Statex \textbf{Input:} A particle of population $S_2$ given by $\mathbf{l}_{d_2}^{(k_2)}=(\underbrace{l_{d_2}^{(k_2)}[1]}_{v_y}, \underbrace{l_{d_2}^{(k_2)}[2], \ldots,l_{d_2}^{(k_2)}[NI+1]}_{\mathcal{P}})^T$
			\Statex \textbf{Output:} Solution to Problem $\mathrm{(P.1)}$ for given $\{v_y,\mathcal{P}\}$
			\Statex \hspace{5mm} \textbf{Initialization:}
			\State  Initialize iteration number $k_1 \gets 1$
			\State  Initialize  random population $S_{1,d_2}$ of size $D_1$ and random PSO velocities
			\State  Evaluate the fitness of the initial population $S_{1,d_2}$ given $\mathbf{l}_{d_2}^{(k_2)}$ using (\ref{eq:fitness_function_s1})
			\State 	Set the local and global best particles and their respective fitness	
			\While{$k_1 \leq K_1$}
			\State Update PSO velocities using  (\ref{eq:pso_velocity_update_s1}) and (\ref{eq:pso_velocity_wall_s1})
			\State Update population $S_{1,d_2}$ using (\ref{eq:pso_update_s1})
			\State Evaluate the fitness of each particle in $S_{1,d_2}$ for given $\mathbf{l}_{d_2}^{(k_2)}$ using (\ref{eq:fitness_function_s1})
			
			\State Update and record best local particles $\mathbf{p}_{\mathrm{best}, d_1, d_2}^{(k_1)}, \forall d_1$ (and their fitness), the best global particle \Statex \hspace{5mm} \hspace{5mm} $\mathbf{p}_{\mathrm{best}, d_2}^{(k_1)}$ (and its fitness), and the worst fitness value $\sigma_h^{\mathrm{max}}$
			\State Set  iteration number $k_1 \gets k_1 + 1$
			\EndWhile
			\State \textbf{return} $\mathbf{p}_{\mathrm{best}, d_2}^{(K_1)}$
			\EndFunction
		\end{algorithmic}
	\end{algorithm}
	\fi
	In this work, we propose to divide the original problem $\mathrm{(P)}$ into two sub-problems, denoted by $\mathrm{(P.1)}$ and $\mathrm{(P.2)}$. In sub-problem $\mathrm{(P.1)}$, we optimize the \ac{uav} formation in the $xz$-plane, $\mathcal{Q}$, for given velocity $v_y$ and communication transmit powers $\mathcal{P}$. Then, in sub-problem $\mathrm{(P.2)}$, we optimize the \ac{uav} velocity $v_y$ and communication transmit powers $\mathcal{P}$ for a given \ac{uav} formation $\mathcal{Q}$. To explore the search spaces of $\mathrm{(P.1)}$ and $\mathrm{(P.2)}$, we design two types of populations (i.e., two species) that interact within a co-evolutionary framework, guiding the global search towards a sub-optimal solution of $\mathrm{(P)}$. In particular, one species, denoted by $S_2$, evolves for a maximum of $K_2$ generations and contains $D_2$ particles, denoted by $\mathbf{l}_{d_2}^{(k_2)}\in \mathbb{R}^{1+NI}, k_2 \in \{1, \ldots, K_2\}, d_2 \in \{1, \ldots, D_2\}$, which explore the search space of sub-problem $\mathrm{(P.2)}$. The other species contains $D_2$ populations of the same type, denoted by $\{S_{1,d_2}\}_{d_2=1}^{d_2=D_2}$, that are associated with each particle of species $S_2$ and evolve for $K_1$ generations. Each population $S_{1,d_2}$ contains $D_1$ particles, denoted by $\mathbf{p}_{d_1,d_2}^{(k_1)}\in \mathbb{R}^{2I}, k_1 \in \{1, \ldots, K_1\}, d_1 \in \{1, \ldots, D_1\}$, that explore the search space of sub-problem  $\mathrm{(P.1)}$. As illustrated in Figure~\ref{fig:diagram}, population $S_2$ evolves over $K_2$ generations. In each generation, the populations $\{S_{1,d_2}\}_{d_2=1}^{d_2=D_2}$ evolve in parallel for $K_1$ generations. Specifically, for a given generation $k_2$ and candidate solution $\{v_y^*, \mathcal{P}^*\}$ represented by particle $\mathbf{l}_{d_2}^{(k_2)}$ from $S_2$, the associated population $S_{1,d_2}$ evolves to optimize \ac{uav} formation $\mathcal{Q}$ in the $xz$-plane. Once $S_{1,d_2}$ converges to a solution, population $S_2$ proceeds to its next generation. This nested evolutionary process is repeated until $S_2$ completes its $K_2$ generations. The proposed co-evolutionary procedure is inherently parallelizable and is summarized in \textbf{Algorithm}~\ref{alg:coevolution}.\par 
To effectively evolve each population type, a dedicated optimizer is required for each species to update its respective population across generations. Motivated by its simplicity and robustness, we adopt non-parameterized \ac{pso}~\cite{non_parameterized_pso} as the \ac{ea}-based optimizer for both species. \ac{pso} is a population-based metaheuristic algorithm that has demonstrated strong performance in solving highly coupled optimization problems with complex, non-convex search spaces. It {exhibits} rapid convergence and is inspired by collective behaviors observed in nature, such as coordinated foraging strategies of bird flocks \cite{pso}. Notably, \ac{pso} has been successfully applied to a variety of fields, including antenna array design~\cite{antenna_array}, \ac{uav}-based path planning~\cite{path_planning}, and scheduling problems~\cite{scheduling}. Additionally, it performs well in approximating global optima for challenging benchmark functions such as Griewank, Rastrigin, Rosenbrock, and Ackley~\cite{performance_pso}, which are known for their complex and highly coupled search spaces. 
In~\cite{coevolution_pso}, a parameterized co-evolutionary \ac{pso} framework was employed to handle constrained optimization problems by tuning penalty parameters. In contrast, we propose a non-parameterized co-evolutionary \ac{pso} framework that leverages co-evolution to decompose and solve sub-problems derived from the original optimization task. This approach reduces the likelihood of convergence to local optima and offers high parallelizability, enabling efficient implementation, as illustrated in Figure~\ref{fig:diagram}. In \textbf{Algorithm} \ref{alg:pso}, we summarize the steps of the {proposed} \ac{pso} algorithm, which is called in each iteration of the co-evolutionary procedure, summarized in \textbf{Algorithm} \ref{alg:coevolution}. In the next sections, we explain how the two species, $S_2$ and $\{S_{1,d_2}\}_{d_2=1}^{d_2=D_2}$, co-evolve using non-parameterized \ac{pso} to minimize the height error $\sigma_h$, subject to all constraints of problem $\mathrm{(P)}$.
	\subsection{Solution to Problem $\mathrm{(P.1)}$ via Populations $\{S_{1,d_2}\}_{d_2=1}^{d_2=D_2}$} 
	In this sub-section, we focus on optimizing the \ac{uav} formation in the $xz$-plane for a given swarm velocity and communication resource allocation. Therefore, problem $\mathrm{(P)}$ reduces to sub-problem $\mathrm{(P.1)}$ given by:
	\begin{alignat*}{2} 
		&\mathrm{(P.1)}:\min_{\mathcal{Q}} \hspace{3mm}  \sigma_{h}  & \qquad&  \\
		\text{s.t.} \hspace{3mm} &\mathrm{C1,C2,C5,C6,C7,C8}. &       &     
	\end{alignat*}
	Sub-problem $\mathrm{(P.1)}$ is non-convex due to the complexity of its objective function and constraints. In fact, the presence of non-linear trigonometric functions significantly increases the complexity of the problem. Additionally, a strong coupling exists among the optimization variables as changing the position of a single slave \ac{uav} affects the perpendicular baseline, the \ac{hoa}, and the height error of $I-1$ \ac{insar} pairs, (see (\ref{eq:perpendicular_baseline}), (\ref{eq:height_of_ambiguity}), and (\ref{eq:height_error})), respectively. Thus, the $2 I$ optimization variables of $\mathrm{(P.1)}$ are coupled, making it difficult to solve the problem based on classical optimization techniques. To address this challenge, we leverage {the} non-parameterized \ac{pso} algorithm \cite{coevolution_pso}, such that  populations $\{S_{1,d_2}\}_{d_2=1}^{d_2=D_2}$ are used to explore the search space of problem $\mathrm{(P.1)}$. In iteration $k_1$ of the \ac{pso} algorithm, population $S_{1,d_2}$ consists of $D_1$ particles of size $2I$, denoted by $\mathbf{p}_{d_1, d_2}^{(k_1)}=(p_{d_1, d_2}^{(k_1)}[1], \ldots,p_{d_1, d_2}^{(k_1)}[2I])^T\in \mathbb{R}^{2I}, d_1 \in \{1, \dots, D_1\}, k_1 \in \{1, \ldots, K_1\}$, and given by: \ifonecolumn
	\begin{equation}
		\mathbf{p}^{(k_1)}_{d_1,d_2}=(\overbrace{x^{(k_1)}_{1,{d_1,d_2}},z^{(k_1)}_{1,{d_1,d_2}}}^{U_1}, \dots,\overbrace{x^{(k_1)}_{I,{d_1,d_2}},z^{(k_1)}_{I,{d_1,d_2}}}^{U_I})^T, \forall d_1, \forall d_2, \forall k_1,
	\end{equation}\else
		\begin{gather}
		\mathbf{p}^{(k_1)}_{d_1,d_2}=(\overbrace{x^{(k_1)}_{1,{d_1,d_2}},z^{(k_1)}_{1,{d_1,d_2}}}^{U_1}, \dots,\overbrace{x^{(k_1)}_{I,{d_1,d_2}},z^{(k_1)}_{I,{d_1,d_2}}}^{U_I})^T, \\ \notag \forall d_1, \forall d_2, \forall k_1,
	\end{gather}
	\fi
	where $x^{(k_1)}_{i,d_1,d_2}$ and $z^{(k_1)}_{i,d_1,d_2}$ correspond to the candidate $x$- and  $z$- positions of $U_i, i \in \mathcal{I}$, respectively, as selected by particle $d_1$ of population $S_{1,d_2}$ in iteration $k_1$. Furthermore, each particle $\mathbf{p}^{(k_1)}_{d_1,d_2}$ of population  $S_{1,d_2}$ changes its position in iteration $k_1$ according to velocity vector $\mathbf{v}^{(k_1)}_{d_1,d_2} \in \mathbb{R}^{2I}$, given by: 
	\begin{equation}\label{eq:PSO_velocity}
		\mathbf{v}_{d_1,d_2}^{(k_1)}=(v^{(k_1)}_{d_1,d_2}[1], \dots, v^{(k_1)}_{d_1,d_2}[2I])^T, \forall d_1, \forall d_2, \forall k_1. 
	\end{equation}
	In iteration $k_1 \geq 2$, the velocity of each particle is affected by its previous velocity, given by $\mathbf{v}_{d_1, d_2}^{(k_1-1)}$, its previous local experience, given by its previous local best-known position, denoted by $\mathbf{p}_{\mathrm{best}, d_1,d_2}^{(k_1-1)}\in \mathbb{R}^{2I}$, and its global experience, given by the best-known position, denoted by $\mathbf{p}_{\mathrm{best},d_2}^{(k_1-1)}\in \mathbb{R}^{2I}$, as shown below \cite{pso}:
	\ifonecolumn
	\begin{equation}\label{eq:pso_velocity_update_s1}
		\mathbf{v}_{d_1, d_2}^{(k_1)}=  \stackrel{\text{Inertial term}}{\overbrace{w^{(k_1-1)} \mathbf{v}_{d_1, d_2}^{(k_1-1)}}} +\stackrel{\text{Cognitive term}}{\overbrace{ c_1 \mathcal{R}_1 (\mathbf{p}_{\mathrm{best}, d_1, d_2}^{(k_1-1)}-\mathbf{p}_{d_1, d_2}^{(k_1-1)} )}} + 	\stackrel{\text{Social term}}{\overbrace{ c_2  \mathcal{R}_2 (\mathbf{p}_{\mathrm{best}, d_2}^{(k_1-1)}-\mathbf{p}_{d_1, d_2}^{(k_1-1)})}}, \forall k_1 \geq2,
	\end{equation} 
	\else
	\begin{align}\label{eq:pso_velocity_update_s1}
		\mathbf{v}_{d_1, d_2}^{(k_1)}=&  \stackrel{\text{Inertial term}}{\overbrace{w^{(k_1-1)} \mathbf{v}_{d_1, d_2}^{(k_1-1)}}} +\stackrel{\text{Cognitive term}}{\overbrace{ c_1 \mathcal{R}_1 (\mathbf{p}_{\mathrm{best}, d_1, d_2}^{(k_1-1)}-\mathbf{p}_{d_1, d_2}^{(k_1-1)} )}} + \notag \\ 	&\stackrel{\text{Social term}}{\overbrace{ c_2  \mathcal{R}_2 (\mathbf{p}_{\mathrm{best}, d_2}^{(k_1-1)}-\mathbf{p}_{d_1, d_2}^{(k_1-1)})}}, \forall k_1 \geq2,
	\end{align} 
	\fi
	where $c_1$ and $c_2$  denote the cognitive and social learning {factors}, respectively, $\mathcal{R}_1$ and $\mathcal{R}_2$ are random variables that are uniformly distributed in $[0,1]$, and $w^{(k_1-1)} \in [0,1]$ is the inertial weight in iteration $k_1-1$. Different  methods to adjust the inertial weight {have been} proposed in the literature  \cite{performance_pso}. In this work, we employ a linearly decaying inertial weight. Moreover, {the} initial velocity vectors, denoted by $\mathbf{v}_{d_1, d_2}^{(1)}$, are uniformly distributed, such that $\mathbf{v}_{d_1, d_2}^{(1)} \sim \mathcal{U}\left(\mathbf{0}_{2I},v_{\rm PSO}^{\rm max}\mathbf{1}_{2I}\right), \forall d_1, \forall d_2$, where $v_{\rm PSO}^{\rm max}$ is the  maximum particle velocity. Note that the \ac{pso} velocity in (\ref{eq:pso_velocity_update_s1}) might {cause} particles {to move} outside of the feasible search space of $\mathrm{(P.1)}$. To prevent {this from happening} and {to} speed up the convergence of the \ac{pso} algorithm, we define boundaries for the search space based on constraints $\mathrm{C1}$ and $\mathrm{C2}$. Based on constraint $\mathrm{C2}$, we can show that: 
	\begin{equation}
		x_{\rm min} \leq x_i \leq x_{\rm max}, \forall i \in \mathcal{I},
	\end{equation}
	where $x_{\rm min}=x_t - z_{\rm max}\tan(\theta_{\rm max})$ and $x_{\rm max}=x_t - z_{\rm min}\tan(\theta_{\rm min})$. Consequently, we define the maximum and minimum boundaries for problem $\mathrm{(P.1)}$ using two particles denoted by $\mathbf{p}_{\rm max}$ and $\mathbf{p}_{\rm min}$, respectively, where $\mathbf{p}_{\rm max}=(p_{\rm max}[1],...,p_{\rm max}[2I])^T\in \mathbb{R}^{2I}$ is given by: 
	\begin{equation}
		p_{\rm max}[2i-1]=x_{\rm max}, p_{\rm max}[2i]=z_{\rm max}, \forall  i \in \mathcal{I}, 
	\end{equation}
	and  $\mathbf{p}_{\rm min}=(p_{\rm min}[1],...,p_{\rm min}[2I])^T\in \mathbb{R}^{2I}$ is defined as follows:
	\begin{equation}
		p_{\rm min}[2i-1]=x_{\rm min}, p_{\rm min}[2i]=z_{\rm min}, \forall  i \in \mathcal{I}.
	\end{equation}
	Based on these boundaries, we implement a reflecting wall mechanism as described in \cite{boundary}. In particular, this method prevents particles from leaving the search space by reflecting them at the defined boundary. To this end, each particle's velocity is adjusted as follows: 
	\ifonecolumn
	\begin{equation}\label{eq:pso_velocity_wall_s1}
		v_{d_1, d_2}^{(k_1)} [i] = - v_{d_1, d_2}^{(k_1)}[i], \text{ if } 
		p_{d_1, d_2}^{(k_1)}[i] + v_{d_1, d_2}^{(k_1)}[i] \notin [p_{\rm min}[i], p_{\rm max}[i]], \quad \forall d_1, \forall d_2, \forall i, \forall k_1.
	\end{equation}
	\else
	\begin{gather}\label{eq:pso_velocity_wall_s1}
	v_{d_1, d_2}^{(k_1)} [i] = - v_{d_1, d_2}^{(k_1)}[i], \text{ if } \\ \notag
	p_{d_1, d_2}^{(k_1)}[i] + v_{d_1, d_2}^{(k_1)}[i] \notin [p_{\rm min}[i], p_{\rm max}[i]],  \forall d_1, \forall d_2, \forall i, \forall k_1.
	\end{gather}
	\fi 
	Initially, the particles of the first generation (i.e., $k_1=1$) are uniformly distributed, such that $\mathbf{p}_{d_1, d_2}^{(1)} \sim \mathcal{U}(\mathbf{p}_{\rm min},\mathbf{p}_{\rm max}), \forall d_1, \forall d_2$. Next, in each iteration $k_1\geq 2$, each particle's position is updated as follows: 
	\begin{equation} \label{eq:pso_update_s1}
		\mathbf{p}_{d_1, d_2}^{(k_1+1)}=\mathbf{p}_{d_1, d_2}^{(k_1)}+\mathbf{v}_{d_1, d_2}^{(k_1)}, \forall d_1, \forall d_2, \forall k_1 \leq K_1-1.
	\end{equation}
	Although the boundary is defined based on constraints $\mathrm{C1}$ and $\mathrm{C2}$, it only enforces the feasibility  of constraint $\mathrm{C1}$. Consequently, constraint $\mathrm{C2}$, along with the remaining constraints of $\mathrm{(P.1)}$, must be incorporated into the fitness function to ensure feasibility. While the \ac{pso} algorithm was initially introduced for unconstrained optimization \cite{pso}, various adaptations have been developed to handle constraints, as seen in \cite{coevolution_pso,non_parameterized_pso,constrained_pso}. In this work, we employ a non-parameterized \ac{pso} approach, originally proposed in \cite{non_parameterized_pso}. Specifically, for given $v_y$ and $\mathcal{P}$, the fitness function of population $S_{1,d_2}$, denoted by $\mathcal{F}_1$, is defined as follows:
	\begin{equation}\label{eq:fitness_function_s1}	
		\mathcal{F}_1(\mathbf{p}_{d_1, d_2}^{(k_1)}|\mathbf{l}_{d_2}^{(k_2)})=\hspace{-1mm}\begin{dcases}
			\sigma_h, &\text{ if } \mathbf{p}_{d_1, d_2}^{(k_1)} \in F_1,\\ 
			\sigma_h^{\rm max}	+\sum\limits_{l \in \mathcal{L}_1}g_l, \hspace{-4mm}&\text{ if } \mathbf{p}_{d_1, d_2}^{(k_1)} \notin F_1,\\
		\end{dcases} \forall d_1, \forall d_2,
	\end{equation}
	where $\sigma_h$ is evaluated at the candidate solution $\mathbf{p}_{d_1, d_2}^{(k_1)}$ for fixed $v_y$ and $\mathcal{P}$, given by vector $\mathbf{l}_{d_2}^{(k_2)}$. $\sigma_h^{\rm max}$ is the worst feasible height error observed across all $d_1$, $d_2$, and  $k_1$,  $F_1$ is the feasible set for sub-problem $\mathrm{(P.1)}$ defined by constraints $\mathrm{C2}$ and $\mathrm{C5-C8}\footnote{Constraint $\mathrm{C1}$ is always satisfied based on the reflecting wall.}$, and functions $g_l, l \in \mathcal{L}_1=\{2,5,6,7,8\},$ are used as metrics to quantify violations of constraints  $\mathrm{C2}$ and $\mathrm{C5-C8}$, respectively. For {the} $v_y$ and $\mathcal{P}$ given by particle $\mathbf{l}_{d_2}^{(k_2)}$, functions $g_l$ are evaluated at the candidate solution $\mathbf{p}_{d_1, d_2}^{(k_1)}$ as follows:
	\begin{align}
		&g_2=\sum\limits_{i=1}^{I}\left(\left[ \left(\theta_{\rm min} - \theta_i\right)\right]^+ +\left[\left(\theta_i-\theta_{\rm max}\right)\right]^+\right), \label{eq:g2}\\
		&g_5= \sum\limits_{i<j}^{I}\left[ d_{\mathrm{min}}-b[i,j]\right]^+,\\
		& g_6=\left[C_{\rm min}-C_{\mathrm{tot}}\right]^+,\\
		&g_7= \sum\limits_{i<j}^{I}\left[h_{\rm amb}^{\rm min}- h_{\mathrm{amb}}[i,j]\right]^+,\\
		&g_8=  \sum\limits_{n=1}^{N}\sum\limits_{i=1}^{I}\left[  R_{\mathrm{min},i}-R_{i,n}\right]^+. \\
	\end{align} 
	\subsection{Solution to Problem $\mathrm{\rm (P.2)}$ via Population $S_2$}
	Next, we focus on optimizing the velocity of the \ac{uav} swarm along the $y$-axis and the communication transmit powers, denoted by  $v_y$ and $\mathcal{P}$, respectively, for a given \ac{uav} formation $\mathcal{Q}$. Thus, problem $\mathrm{(P)}$ reduces to sub-problem $\mathrm{(P.2)}$ given by:
	\begin{alignat*}{2} 
		&\mathrm{(P.2)}:\min_{v_y,\mathcal{P}} \hspace{3mm}  \sigma_{h}   & \qquad&  \\
		\text{s.t.} \hspace{3mm} &\mathrm{C3,C4,C8,C9}. &       &     
	\end{alignat*}
	Problem $\mathrm{(P.2)}$ is non-convex due to the objective function and constraint $\mathrm{C8}$. {To solve  $\mathrm{(P.2)}$}, we use a second type of population, denoted by $S_{2}$, which consists of $D_2$ particles of size $NI+1$, denoted by $\mathbf{l}_{d_2}^{(k_2)}=(l_{d_2}^{(k_2)}[1], \ldots,l_{d_2}^{(k_2)}[NI+1])^T\in \mathbb{R}^{NI+1}, d_2 \in \{1, \dots, D_2\}, k_2 \in \{1, \ldots, K_2\}$. Each particle encodes optimization variables $v_y$ and $\mathcal{P}$ as follows:
	\ifonecolumn
	\begin{equation}
		\mathbf{l}_{d_2}^{(k_2)}=\left(\underbrace{l_{d_2}^{(k_2)}[1]}_{v_y},\underbrace{l_{d_2}^{(k_2)}[2],\ldots,l_{d_2}^{(k_2)}[N+1]}_{\mathbf{P}_{\mathrm{com},1}},\ldots,\underbrace{l_{d_2}^{(k_2)}[(I-1)N+2],\ldots,l_{d_2}^{(k_2)}[NI+1]}_{\mathbf{P}_{\mathrm{com},I}}\right)^T, \forall d_2, \forall k_2.
	\end{equation}
	\else
		\begin{gather}
		\mathbf{l}_{d_2}^{(k_2)}=\Bigg(\underbrace{l_{d_2}^{(k_2)}[1]}_{v_y},\underbrace{l_{d_2}^{(k_2)}[2],\ldots,l_{d_2}^{(k_2)}[N+1]}_{\mathbf{P}_{\mathrm{com},1}},\ldots,\notag \\ \underbrace{l_{d_2}^{(k_2)}[(I-1)N+2],\ldots, l_{d_2}^{(k_2)}[NI+1]}_{\mathbf{P}_{\mathrm{com},I}}\Bigg)^T, \forall d_2, \forall k_2.
	\end{gather}
	\fi Similar to populations $\{S_{1,d_2}\}_{d_2=1}^{d_2=D_2}$, we employ a non-parameterized \ac{pso} optimizer to evolve population $S_2$. However, the search space of problem $\mathrm{(P.2)}$ can be significantly reduced by noting that $\sigma_h$ is independent of $\mathcal{P}$, therefore, solving $\mathrm{(P.2)}$ for a given (i.e., fixed) $v_y$ is equivalent to solving the following feasibility problem:
	\begin{alignat*}{2} 
		&\mathrm{(P.2)'}:\min_{\mathcal{P}} \hspace{3mm}  1   & \qquad&  \\
		\text{s.t.} \hspace{3mm} &\mathrm{C4,C8,C9}. &       &     
	\end{alignat*}
	Thus, the \ac{pso} particles optimize the swarm velocity $v_y$ while ensuring that constraint $\mathrm{C3}$ is satisfied (which acts as a boundary condition for population $S_2$), in addition to solving the feasibility problem $\mathrm{(P.2)'}$.
	\begin{proposition}\label{prop:solution_communication_power}
		For a given \ac{uav} velocity $v_y$ and \ac{uav} formation $\mathcal{Q}$, problem $\mathrm{(P.2)'}$ is feasible if and only if the following equivalent constraints (which are independent of $\mathcal{P}$) are feasible: 
		\begin{align}\label{eq:communication_power}
			\begin{dcases}
				\mathrm{C10:} \eta_{i,n} &\leq P_{\rm com}^{\rm max}, \forall n \in \mathcal{N}, \forall i\in \mathcal{I},\\
				\mathrm{C11:} \sum_{n=1}^{N} \eta_{i,n} &\leq \frac{E_{\rm max}}{\delta_t} - N P_{\rm prop}-N P_{\mathrm{rad},i}, \forall i\in \mathcal{I},
			\end{dcases}
		\end{align}
		where $\eta_{i,n}= \frac{d^2_{i,n}}{\beta_{c,i}}(2^{\frac{R\mathrm{min},i}{B_{c,i}}}-1)$. In case problem  $\mathrm{(P.2)'}$ is feasible, the communication power allocation {can be chosen as} $P^*_{\mathrm{com},i}[n]=\eta_{i,n}, \forall i, \forall n$. 
	\end{proposition}
	\begin{proof}
		Constraints $\mathrm{C4}$, $\mathrm{C8}$, and $\mathrm{C9}$ can be reformulated as the following system of inequalities:
		\begin{align}\label{proof:system1}
			\begin{dcases}
				&P_{\mathrm{com},i}[n] \leq P^{\mathrm{max}}_{\mathrm{com}}, \forall n, \forall i,\\
				&P_{\mathrm{com},i}[n] \geq \eta_{i,n}, \forall n, \forall i,  \\
				&\sum\limits_{n=1}^{N} P_{\mathrm{com},i}[n] \leq \frac{E_{\rm max}}{\delta_t} - N P_{\rm prop} - N P_{\mathrm{rad},i}, \forall i,
			\end{dcases}
		\end{align}
		where $\eta_{i,n} = \frac{d^2_{i,n}}{\beta_{c,i}} \left(2^{\frac{R^{\mathrm{min,i}}}{B_{c,i}}} - 1\right) \geq 0, \forall i, \forall n$. The system of inequalities in (\ref{proof:system1}) describes a simple communication power allocation task for each \ac{uav} over time. The minimum and maximum instantaneous communication transmit powers of $U_i$ required in time slot $n$ are given by $\eta_{i,n}$ and $P^{\mathrm{max}}_{\mathrm{com}}$, respectively, and the total communication power consumed over time must not exceed the limit $\frac{E_{\rm max}}{\delta_t} - N P_{\rm prop} - N P_{\mathrm{rad},i}$ for each drone. A trivial solution to this resource allocation problem is to set the communication power in each time slot to its minimum required level, i.e., $P^*_{\mathrm{com},i}[n] = \eta_{i,n}, \forall i, \forall n$. Thus, the system of inequalities (\ref{proof:system1}), i.e., constraints $\mathrm{C4}$, $\mathrm{C8}$, and $\mathrm{C9}$, is feasible if and only if the system of inequalities (\ref{proof:system1}) is satisfied.
		In a nutshell,  rather than exploring the full $NI$-dimensional search space of problem $\mathrm{(P.2)'}$, it is sufficient to check whether values $\eta_{i,n}$ satisfy constraints $\mathrm{C10}$ and $\mathrm{C11}$ for all $i$ and $n$. If feasible, these values can be directly used as the solution to problem $\mathrm{(P.2)'}$, i.e., by setting $P^*_{\mathrm{com},i}[n] = \eta_{i,n}$ for all $i$ and $n$.
	\end{proof}
	Based on \textbf{Proposition}~\ref{prop:solution_communication_power}, for a given swarm velocity $v_y$, represented by $l_{d_2}^{(k_2)}[1]$, the communication transmit powers do not influence the objective function. Instead, they are determined by solving feasibility problem $\mathrm{(P.2)'}$. As a result, the velocity-based \ac{pso} update process is reduced to evolving a single variable, $v_y$, rather than $NI+1$ variables. The transmit powers, originally represented by $(l_{d_2}^{(k_2)}[2], \ldots, l_{d_2}^{(k_2)}[NI+1])^T$, are now set for a given $v_y=l_{d_2}^{(k_2)}[1]$ according to $P^*_{\mathrm{com},i}[n] = \eta_{i,n}$, and are used solely to compute penalty terms associated with constraints $\mathrm{C10}$ and $\mathrm{C11}$ in the fitness evaluation. The \ac{pso} velocity used to evolve $l_{d_2}^{(k_2)}[1]$ in population $S_2$ is given by:
	\ifonecolumn
	\begin{equation}\label{eq:pso_velocity_update_s2}
		v_{\mathrm{PSO}, d_2}^{(k_2)}= \overbrace{w^{(k_2-1)} v_{\mathrm{PSO}, d_2}^{(k_2-1)}}^{\text{Inertial weight}} + \overbrace{c_1 \mathcal{R}_1 (l_{\mathrm{best}, d_2}^{(k_2-1)}-l_{d_2}^{(k_2-1)}[1] )}^{\text{Cognitive term}} +
		\overbrace{c_2  \mathcal{R}_2 (l_{\mathrm{best}}^{(k_2-1)}-l_{d_2}^{(k_2-1)}[1])}^{\text{Social term}}, \forall k_2 \geq2, \forall d_2,
	\end{equation} 
	\else
	\begin{align}\label{eq:pso_velocity_update_s2}
		v_{\mathrm{PSO}, d_2}^{(k_2)}=& \overbrace{w^{(k_2-1)} v_{\mathrm{PSO}, d_2}^{(k_2-1)}}^{\text{Inertial weight}} + \overbrace{c_1 \mathcal{R}_1 (l_{\mathrm{best}, d_2}^{(k_2-1)}-l_{d_2}^{(k_2-1)}[1] )}^{\text{Cognitive term}} +\notag \\ &\overbrace{c_2  \mathcal{R}_2 (l_{\mathrm{best}}^{(k_2-1)}-l_{d_2}^{(k_2-1)}[1])}^{\text{Social term}}, \forall k_2 \geq2, \forall d_2,
	\end{align} 
	\fi
	where $l_{\mathrm{best}}^{(k_2-1)} \in \mathbb{R}$ is the previous global best-known velocity for the \ac{uav} swarm\footnote{The reader should not confuse the \ac{pso} velocity, denoted by $v_{\mathrm{PSO}, d_2}^{(k_1)}$, with the \ac{uav} swarm velocity $v_y$, which is represented by component $l_{d_2}^{(k_2)}[1]$ of the $d_2$-th particle in population $S_2$ at iteration $k_2$.} and $l_{\mathrm{best}, d_2}^{(k_2-1)} \in \mathbb{R}$ is the previous local best-known velocity of the \ac{uav} swarm for particle $d_2$.  Similar to (\ref{eq:pso_velocity_wall_s1}), constraint $\mathrm{C3}$ of problem $\mathrm{(P)}$ can be handled by implementing a reflecting wall as follows \cite{boundary}: 
	\ifonecolumn
	\begin{equation}\label{eq:pso_velocity_wall_s2}
		v_{\mathrm{PSO}, d_2}^{(k_2)}= - v_{\mathrm{PSO}, d_2}^{(k_2)}, \text{ if } 
		l_{d_2}^{(k_2)}[1] + v_{\mathrm{PSO}, d_2}^{(k_2)} \notin [v_{\rm min}, v_{\rm max}], \quad \forall d_2, \forall k_2 \geq 2.
	\end{equation}
	\else
	\begin{gather}
		\label{eq:pso_velocity_wall_s2}
		v_{\mathrm{PSO}, d_2}^{(k_2)}= - v_{\mathrm{PSO}, d_2}^{(k_2)}, \text{ if } \\ \notag
		l_{d_2}^{(k_2)}[1] + v_{\mathrm{PSO}, d_2}^{(k_2)} \notin [v_{\rm min}, v_{\rm max}], \quad \forall d_2, \forall k_2 \geq 2.
	\end{gather}
	\fi 
	In each iteration $k_2 \geq 2$, the component $l_{d_2}^{(k_2)}[1]$, representing $v_y$, is updated as follows:
	\begin{equation}\label{eq:update_population_s2}
		l_{d_2}^{(k_2)}[1] = l_{d_2}^{(k_2-1)}[1] + v_{\mathrm{PSO}, d_2}^{(k_2)}, \quad \forall d_2, \forall k_2 \geq 2.
	\end{equation}
	Initially, $l_{d_2}^{(1)}[1]$ is randomly selected from the interval $[v_{\mathrm{min}}, v_{\mathrm{max}}]$. Furthermore, in each iteration $k_2$, the remaining components of the particle, $\left(l_{d_2}^{(k_2)}[2], \ldots, l_{d_2}^{(k_2)}[NI+1]\right)^T$, representing the communication transmit power allocation, are updated based on \textbf{Proposition}~\ref{prop:solution_communication_power}, such that $P_{\mathrm{com},i}[n] = \eta_{i,n}, \forall i, \forall n$, given $v_y = l_{d_2}^{(k_2)}[1]$.
	
	The remaining constraints of problem $\mathrm{(P.2)}$ (i.e., $\mathrm{C4}$, $\mathrm{C8}$, and $\mathrm{C9}$), expressed equivalently by $\mathrm{C10}$ and $\mathrm{C11}$ according to \textbf{Proposition}~\ref{prop:solution_communication_power}, are incorporated into the fitness function as follows \cite{non_parameterized_pso}:
	\ifonecolumn
	\begin{equation}\label{eq:fitness_s2}
		\mathcal{F}_2(\mathbf{l}_{d_2}^{(k_2)}| \mathbf{p}_{\mathrm{best},d_2}^{(K_1)})=\begin{dcases}
			\mathcal{F}_1(\mathbf{p}_{\mathrm{best},d_2}^{(K_1)}),\hspace{5mm}&\text{if } \mathbf{l}_{d_2}^{(k_2)} \in F_2,\\
			\sigma_h^{\mathrm{max}} + \sum\limits_{l_2 \in \mathcal{L}_2} g_{l_2},\hspace{5mm}&\text{if } \mathbf{l}_{d_2}^{(k_2)} \notin F_2,\\
		\end{dcases} \hspace{5mm}\forall d_2,
	\end{equation}
	\else 
		\begin{equation}\label{eq:fitness_s2}
		\mathcal{F}_2(\mathbf{l}_{d_2}^{(k_2)}| \mathbf{p}_{\mathrm{best},d_2}^{(K_1)})=\begin{dcases}
			\mathcal{F}_1(\mathbf{p}_{\mathrm{best},d_2}^{(K_1)}),&\text{if } \mathbf{l}_{d_2}^{(k_2)} \in F_2,\\
			\sigma_h^{\mathrm{max}} + \sum\limits_{l_2 \in \mathcal{L}_2} g_{l_2},&\text{if } \mathbf{l}_{d_2}^{(k_2)} \notin F_2,\\
		\end{dcases} \forall d_2,
	\end{equation}
	\fi
	where the set $F_2$ defines the feasible region of sub-problem $\mathrm{(P.2)}$ and functions $g_{l_2}$,  $l_2 \in \mathcal{L}_2 = \{10, 11\}$, serve as penalty terms quantifying the degree of constraint violation:
	\begin{align}
		&g_{10}=\sum\limits_{n=1}^{N}\sum\limits_{i=1}^{I}\left[\eta_{i,n} - P_{\rm com}^{\rm max}\right]^+, \label{eq:g10}\\
		&g_{11}=\sum_{i=1}^{I}\sum_{n=1}^{N}\left[ \eta_{i,n}-\frac{E_{\rm max}}{\delta_t} + N P_{\rm prop}+N P_{\mathrm{rad},i}\right]^+.
	\end{align} 
	\subsection{Complexity and Convergence Analysis}
	 The computational complexity of the proposed solution outlined in \textbf{Algorithm}~\ref{alg:coevolution} depends on the  complexity of the underlying \textbf{Algorithm}~\ref{alg:pso}. This is because, each particle in population $S_2$ invokes a full run of \textbf{Algorithm}~\ref{alg:pso}. Therefore, to analyze the computational complexity of the proposed solution in  \textbf{Algorithm}~\ref{alg:coevolution}, we begin by evaluating the complexity of \textbf{Algorithm}~\ref{alg:pso}, {denoted by $C_2$}. \textbf{Algorithm} \ref{alg:pso} operates on a population of size $D_1$, with each particle having a dimensionality $2I$, {and its complexity is mainly related to evaluating each particle and updating the population. Let $C_{\mathcal{F}_1}$ denote the computational cost of evaluating fitness function $\mathcal{F}_1$. Then, the total complexity of \textbf{Algorithm}~\ref{alg:pso} is given by $C_2=\mathcal{O}\left(\underbrace{K_1 D_1 C_{\mathcal{F}_1}}_{\text{Fitness evaluation}} + \underbrace{K_1 D_1 I}_{\text{Updating population}}\right)$.} The computational cost {of} evaluating the fitness function $C_{\mathcal{F}_1}$ arises from two sources:  
	(i) \textit{formation-only-dependent constraints}, i.e., constraints enforced in the across-track plane (e.g., $\mathrm{C2}$, and $\mathrm{C5}-\mathrm{C7}$), which scale as $\mathcal{O}(I)$ and 
	(ii) \textit{formation-and-time-dependent constraints}, i.e., constraints enforced across both the across- and along-track planes (e.g., constraint $\mathrm{C8}$), which scale as $\mathcal{O}(N I)$. Therefore, the cost of evaluating the fitness function is $C_{\mathcal{F}_1} = \mathcal{O}(D_1 N I)$, and the overall computational complexity of \textbf{Algorithm}~\ref{alg:pso} becomes $C_2=\mathcal{O}(K_1 D_1 N I)$, which is polynomial in the problem parameters. Let $C_{\mathcal{F}_2}$ denote the complexity of evaluating the fitness function $\mathcal{F}_2$, then, the total complexity { of \textbf{Algorithm}~\ref{alg:coevolution} is given by $\mathcal{O}\left(\underbrace{K_2 D_2 C_{\mathcal{F}_2} C_2}_{\text{Fitness evaluation}}+ \underbrace{K_2 D_2(NI + 1)}_{\text{Updating population}}\right)$}. Given that the fitness function $\mathcal{F}_2$ involves evaluating a set of constraints that scale with $NI$, we have $C_{\mathcal{F}_2} = \mathcal{O}(NI)$. Therefore, the overall \textit{worst-case serial complexity} {of \textbf{Algorithm}~\ref{alg:coevolution} becomes $\mathcal{O}(K_2 K_1 D_2 D_1 N I)$.} Additionally, the outer loop is parallelizable across $D_2$ and assuming an ideal parallelization with $X = D_2$ cores (see Figure~\ref{fig:diagram}), the \textit{parallel complexity} reduces to $\mathcal{O}(K_2 K_1 D_1 N I)$ which is polynomial in all relevant problem dimensions.\footnote{Further computational speed-up can be achieved through vectorized operations across $N$, $I$, and $D_1$, using tensor-based implementations.}\par 
	The best global value in \textbf{Algorithm}~\ref{alg:coevolution} is guaranteed to asymptotically converge as the number of iterations $K_2 \to \infty$. This convergence is ensured {by} Lines~\ref{line:best_global1} and \ref{line:best_global2}, which continuously track and retain the best solution, thereby, enforcing a non-increasing sequence of global best values. Moreover, it holds that $
	\mathcal{F}_2(\mathbf{l}_{d_2}^{(k_2)} \mid \mathbf{p}_{\mathrm{best},d_2}^{K_1}) \geq 0, \quad \forall \mathbf{l}_{d_2}^{(k_2)}, \ \forall \mathbf{p}_{\mathrm{best},d_2}^{K_1},$ which theoretically proves the asymptotic convergence.
	However, it must be acknowledged that, in theory, convergence alone does not imply optimality nor feasibility. Specifically, in highly constrained or tightly bounded search spaces, the algorithm may converge to sub-optimal or even infeasible solutions. This limitation is not unique to the proposed method but {is} rather a well-known characteristic of most stochastic optimization algorithms designed for non-linear constrained optimization \cite{evolutionary_algorithms,evolutionary_algorithms2}. Nevertheless, in practice, such methods have proven effective in locating not only feasible solutions but often globally optimal ones, as demonstrated on high-dimensional, complex, and strongly coupled benchmark functions~\cite{performance_pso}. 
	In particular, the proposed co-evolutionary algorithm benefits from enhanced exploration capabilities by decomposing the search space and exploring sub-regions of the search space in parallel in each generation. This structural advantage significantly improves {its} ability to escape local minima and to identify high-quality, feasible solutions, as supported by the extensive numerical simulations presented in the next section.
	\section{Simulation Results}\label{sec:simulation_results}
	\begin{table}[]
		\centering
		\caption{System parameters \cite{victor2,snr_equation,coherence1}.}
		\label{tab:my-table}
		\begin{adjustbox}{max width=\columnwidth}
			\begin{tabular}{|c|c?c|c?c|c|}
\hline
Parameter           & Value 					& Parameter & Value 						& Parameter &Value \\ \hline
$N$               &$200$ 					&$h_{\rm amb}^{\rm min}$   &1.2 m  					&$f_0$      &2.5 GHz \\ \hline		
$\delta_t$ 		  &1 s    				&$\gamma_{\rm other}[i,j], \forall i<j$   &0.6			&$\theta_{3\mathrm{dB}}$  &40°    \\\hline
$z_{\mathrm{min}}$&  1 m                  			&$B_{c,i}, \forall i$    & 1 GHz				&$\theta_{\rm min}$  & 37° \\ \hline
$z_{\mathrm{max}}$&  100 m   			  		&$\beta_{c,i}, \forall i$    &20 dB						&$\theta_{\rm max}$  &48.7°\\ \hline
$v_{\mathrm{min}}$& 1 m/s    					&$\sigma_0$           &-10 dB    				&$n_L$   &4  \\ \hline
$v_{\rm max}$ 	  &12 m/s     					&$P_{\mathrm{rad},1}$   &	$15$ dBm		&$n_B$  &4\\ \hline
$x_t$       &20 m 						&$G_t$ &	5 dBi	   					&$U_{\mathrm{tip}}$     &120 m/s   \\ \hline
$g_x $        & 70 m   						&$G_r$ & 5 dBi 	 						&$W_u$  &120   	  \\ \hline
$g_y$   	&150 m      	 				&$\lambda$  &0.12 m    						&$\rho$  & 1.225 kg/m$^3$	\\ \hline
$g_z$  	& 25 m							&$\tau_p\times\mathrm{PRF}$     &10$^{-4}$			&$s$   & 500	   \\ \hline
$d_{\mathrm{min}}$&4 m 				 		&$T_{\mathrm{sys}}$			&400 K			&$A_e$ & 128   \\ \hline
$C_{\rm  min}$   & 4.5$\times10^{4}$ m$^2$			&$B_{\mathrm{Rg}}$  &3 GHz				   	&$\delta_u$ & 0.012\\ \hline
$P_{\mathrm{com}}^{\mathrm{max}}$&39 dBm				&$F$ & 5 dB 							&$\Omega$ & 300 rad/s\\ \hline
$E_{\mathrm{max}}$& 83.33 Wh					&$L$    &  6 dB 						&$R$ &0.4 \\ \hline
			\end{tabular}
		\end{adjustbox} \vspace{-4mm}
	\end{table}
	This section presents simulation results for the co-evolutionary \ac{pso} algorithm outlined in \textbf{Algorithm}~\ref{alg:coevolution}. Unless otherwise specified, the parameters used are as listed in Table~\ref{tab:my-table}. The simulations were performed on a high-performance cluster using two AMD EPYC 7502 processors, each featuring 64 cores. This configuration enables full parallelization of \textbf{Algorithm} \ref{alg:coevolution}, utilizing a total of $X = D_2 = 128$ cores. The algorithm was run for $D_1$=500, a maximum of $K_1 = 500$ generations for \textbf{Algorithm}~\ref{alg:pso} and $K_2 = 100$ generations for the co-evolutionary algorithm. The learning rates were set to $c_1 = 2.5$ and $c_2 = 2$, and the maximum \ac{pso} velocity was constrained to $v_{\mathrm{PSO}}^{\mathrm{max}} = 1$.
	\subsection{Benchmark Schemes}
	To evaluate performance, we adopt the {following} state-of-the-art algorithms as benchmark schemes: 
	\subsubsection{\ac{cga} \cite{cga}}  \Acp{ga} are population-based metaheuristics inspired by natural selection. While traditional \acp{ga} use binary encoding, \ac{cga}s are better suited for real-valued optimization, often yielding improved performance in continuous domains \cite{cga}. This benchmark scheme uses \ac{cga} to solve problem~$\mathrm{(P)}$, incorporating natural selection with a selection rate of 0.3, Gaussian mutation with a mutation rate of 0.1, blend crossover (\(\alpha_b\)-BLX) with \(\alpha_b = 0.1\), and a maximum of 300 generations. To handle the constraints associated with problem~$\mathrm{(P)}$, we adopt a non-parametric penalty function similar to that used in~(\ref{eq:fitness_s2}). The penalty factors are computed based on the communication power allocation derived from \textbf{Proposition}~\ref{prop:solution_communication_power}, ensuring constraint violations are appropriately penalized during fitness evaluation. For brevity, we omit the implementation details of the \ac{cga} and refer interested readers to~\cite{cga} for a comprehensive description.
\subsubsection{\ac{genocop} \cite{genocopII}}\Ac{genocop} is a hybrid \ac{ea} for solving constrained optimization problems, based on the original Genocop I algorithm \cite{genocopI}. While Genocop I handles only linear constraints using convex set theory and genetic search, \ac{genocop} {also} supports nonlinear and non-convex constraints with specialized evolutionary operators. Unlike other methods, \ac{genocop} leverages explicit constraint structure knowledge and incorporates custom mutation and crossover operators to maintain feasibility. Though more computationally intensive than traditional \acp{ga}, these operators enhance performance in tightly constrained, high-dimensional problems.  Our implementation uses an initial temperature of \(10^5\), uniform mutation rate of 0.2, non-uniform mutation rate of 0.7, boundary mutation rate of 0.7, a population size of 100, and a maximum of 50 generations for Genocop II and 300 for Genocop I. For a detailed description of the implementation of Genocop I and II, readers are referred to \cite{genocopI} and \cite{genocopII}, respectively.
\subsubsection{Simulated Annealing (SA) \cite{sa}}  
\ac{sa} is a stochastic global optimization algorithm designed to escape local optima by incorporating randomness into the search process. This feature makes it particularly effective for optimizing nonlinear objective functions where traditional local search algorithms may fail \cite{sa}. The algorithm probabilistically accepts worse solutions, with the likelihood of acceptance governed by a temperature parameter. As the search progresses, the temperature decreases, reducing the probability of accepting inferior solutions and allowing the algorithm to refine the solution toward the global optimum. In this benchmark, we employ a fast simulated annealing schedule with an initial temperature of \( T_0 = 10 \). The temperature {in} each iteration $k_3 \in \{1, \ldots, K_3\}$ is updated according to the following equation:  
\[
T_{k_3} = \frac{T_0}{k_3 + 1},
\]
where $K_3=5\times 10^{3}$. For further implementation details, interested readers are referred to \cite{sa}.
\subsubsection{Deep Reinforcement Learning (DRL) \cite{ddpg}}
 In this benchmark, we implement the \ac{ddpg} algorithm, a model-free, off-policy, actor-critic method for continuous control, inspired by Deep Q-Networks (DQN) \cite{ddpg}. To solve problem \(\mathrm{(P)}\), we model it as a Markov Decision Process (MDP). At each time step \(t\), {state \(\mathbf{s}_t\) and action \(\mathbf{a}_t\) are defined as \(\mathbf{s}_t = \left(x_1[t], z_1[t], \dots, x_I[t], z_I[t], v_y[t]\right)^T\) and \(\mathbf{a}_t = \left(\Delta x_1[t], \Delta z_1[t], \dots, \Delta x_I[t], \Delta z_I[t], \Delta v_y[t]\right)^T\), respectively,} where each \(\Delta\) term represents the change from the previous time step (e.g., \(\Delta v_y[t] = v_y[t] - v_y[t-1]\)). The communication power allocation at each time step is derived from \textbf{Proposition} \ref{prop:solution_communication_power}, and the reward is defined as:  
\[
r_t = \sigma_h - \sum_{l \in \mathcal{L}_1 \cup \mathcal{L}_2} g_l,
\]
where \(\sigma_h\) and \(g_l\) are calculated based on state \(\mathbf{s}_t\) and {the} communication powers. The most important hyperparameters used in our simulations are: Actor learning rate \(10^{-3}\), critic learning rate \(10^{-4}\), and target network soft update rate of 0.3. For full algorithmic details, we refer readers to the {original \ac{ddpg}} paper \cite{ddpg}.
\subsection{Convergence Analysis}
\begin{figure}
	\centering
	\ifonecolumn
		\includegraphics[width=4in]{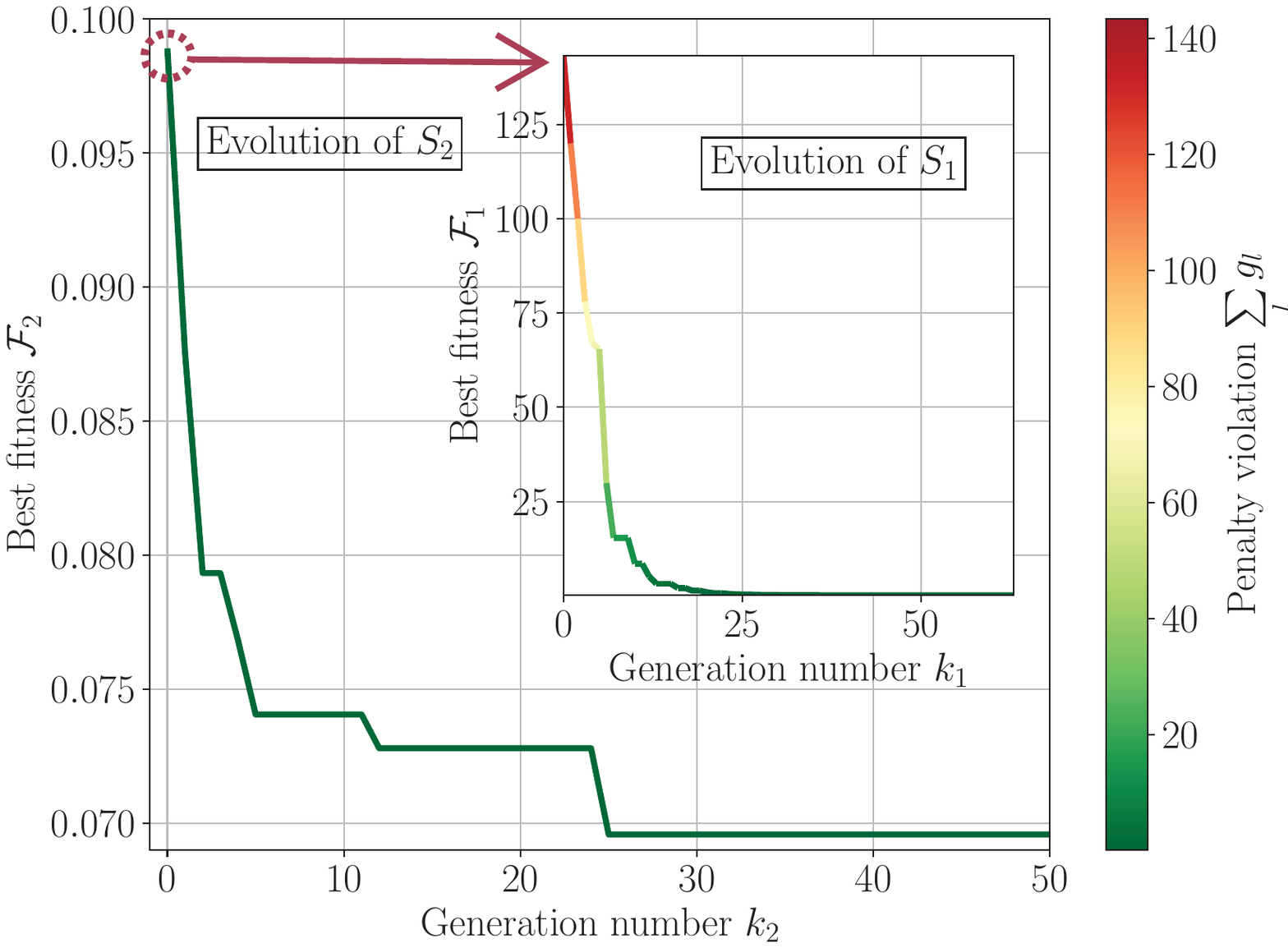}
	\else
		\includegraphics[width=0.9\columnwidth]{figures/proposed_convergence.pdf}
		\fi
\caption{Convergence of the proposed co-evolutionary PSO algorithm. The colormap indicates constraint feasibility for problem $\mathrm{(P)}$ where dark green color refers to feasible solutions. The number of generations shown is limited for clearer visualization.}
\label{fig:proposed_convergence}
	\end{figure} 

\begin{figure*}[t]
	\centering
	\begin{subfigure}[t]{0.24\textwidth}
		\includegraphics[width=\linewidth]{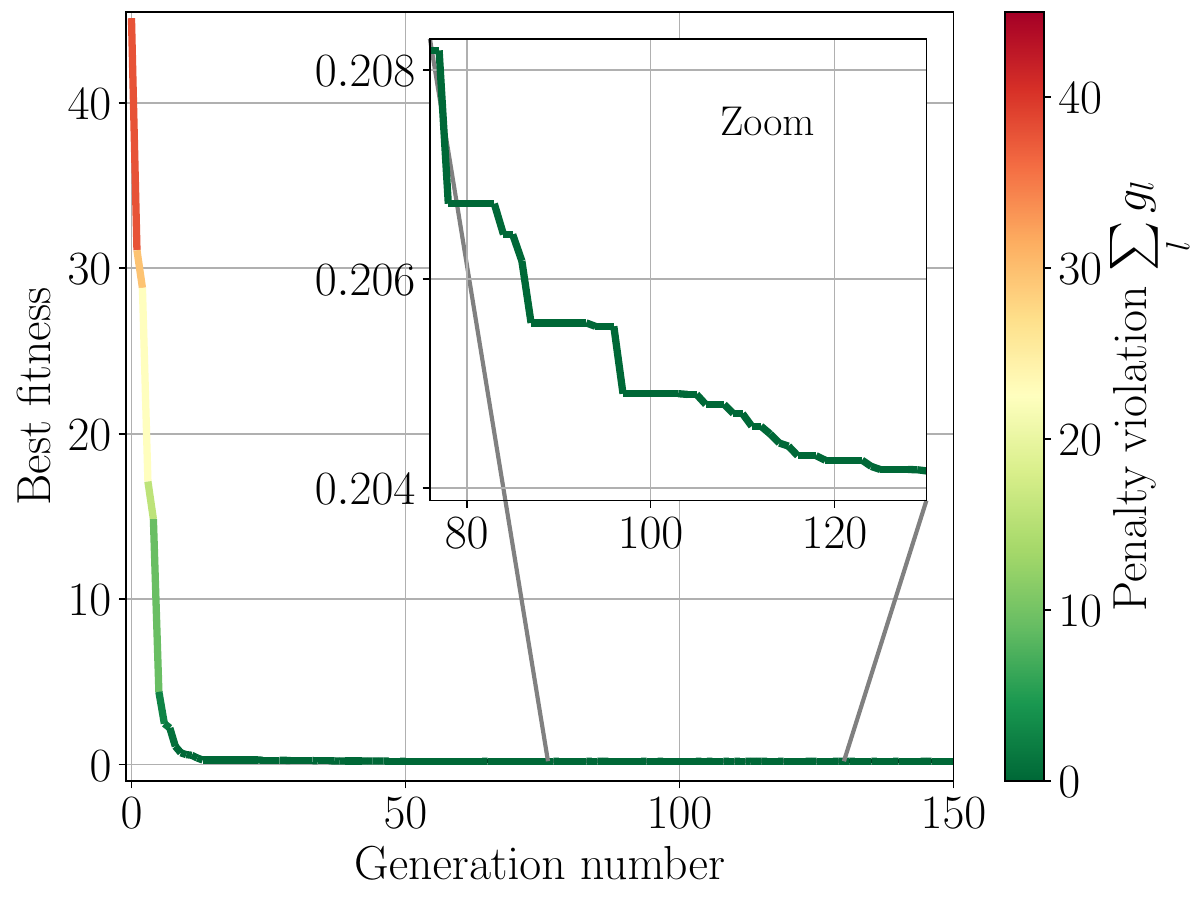}
		\caption{}
	\end{subfigure}
	\hfill
	\begin{subfigure}[t]{0.24\textwidth}
		\includegraphics[width=\linewidth]{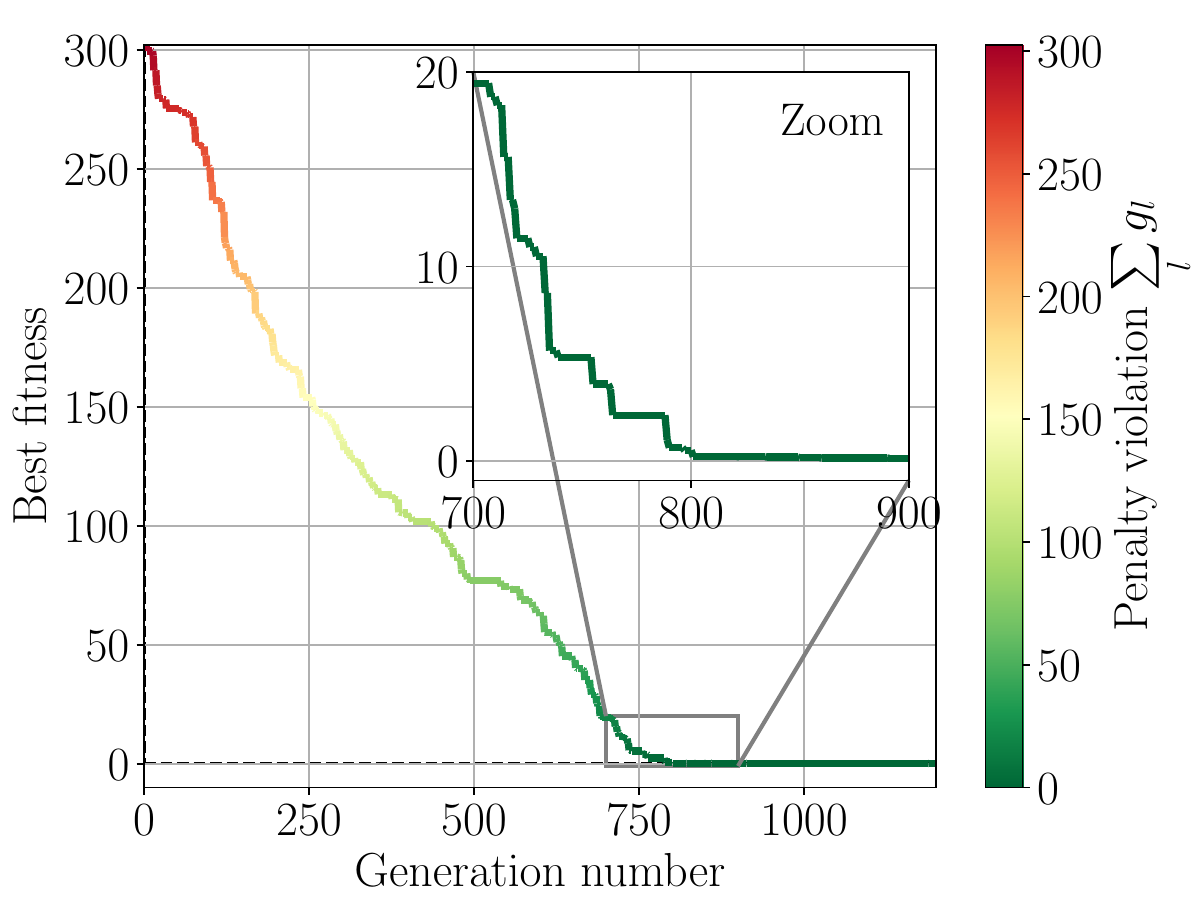}
		\caption{}
	\end{subfigure}
	\hfill
	\begin{subfigure}[t]{0.24\textwidth}
		\includegraphics[width=\linewidth]{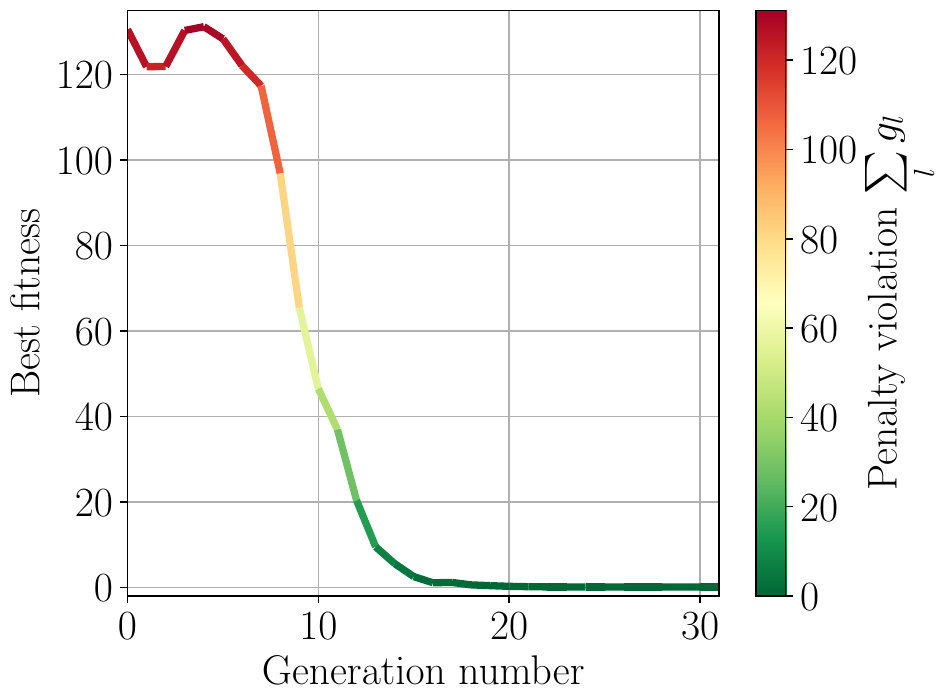}
		\caption{}
	\end{subfigure}
	\hfill
	\begin{subfigure}[t]{0.24\textwidth}
		\includegraphics[width=\linewidth]{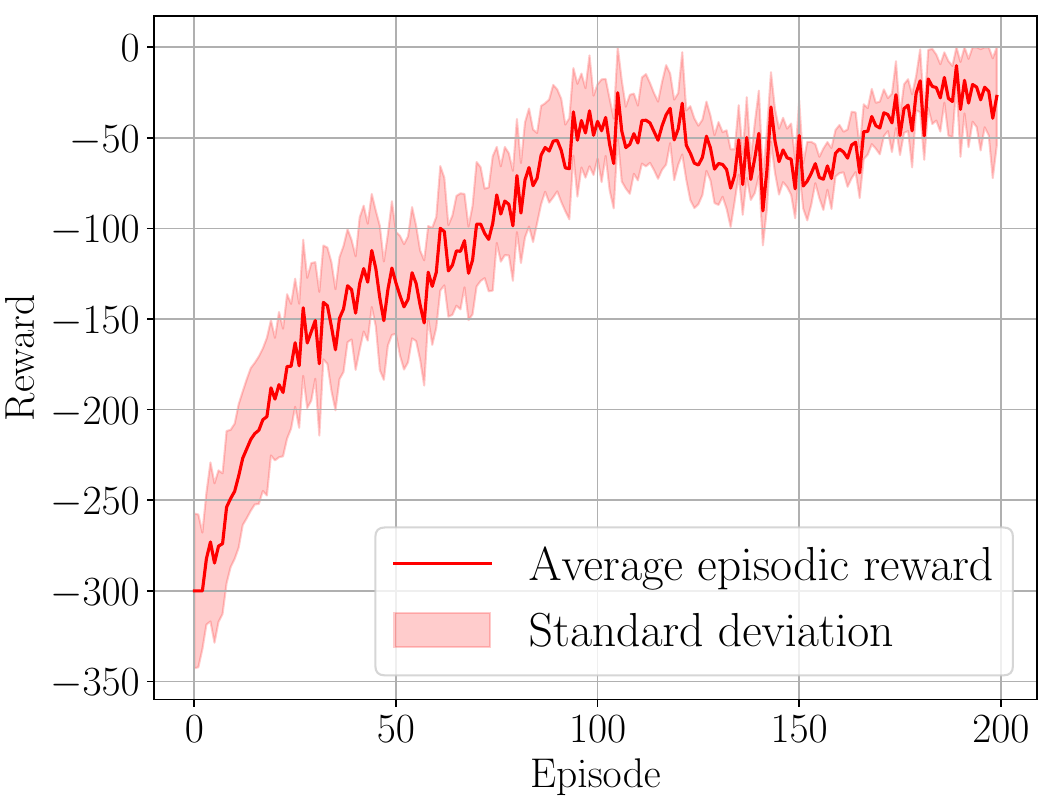}
		\caption{}
	\end{subfigure}
	\caption{Convergence of the different benchmark schemes: (a) CGA, (b) SA, (c) Genocop II, and (d) DRL}
	\label{fig:convergence_benchmark_schemes}
\end{figure*}
In Figures \ref{fig:proposed_convergence} and \ref{fig:convergence_benchmark_schemes}, we examine the convergence behavior of the proposed co-evolutionary algorithm and the different benchmark schemes. The fitness function is plotted against the generation number, with a colormap used to track solution feasibility. Dark green indicates feasible solutions, while red highlights areas with significant constraint violations. Notably, when a solution is feasible, the fitness function aligns with the height error $\sigma_h$, as defined in (\ref{eq:fitness_function_s1}) and (\ref{eq:fitness_s2}). Figure \ref{fig:proposed_convergence} illustrates the improvement in the fitness value as population $S_2$ evolves, based on \textbf{Algorithm} \ref{alg:coevolution}. A zoomed-in plot further shows the evolution of population $S_{1,d_1}$ via \ac{pso} \textbf{Algorithm} \ref{alg:pso}. Initially, the algorithm explores a larger solution space, with constraints not fully satisfied, before converging to a feasible solution after 25 generations. Thanks to the co-evolutionary framework, \textbf{Algorithm} \ref{alg:coevolution} identifies high-quality solutions even early on by evolving in parallel populations $\{S_{1,d_2}\}_{d_1=1}^{d_2=D_2}$ {for} $K_1$ generations. For the {considered} system configuration, with a network size of $I = 5$, the co-evolutionary algorithm achieves a height accuracy of up to 7 cm. This performance is achieved with an optimized UAV swarm speed of 4.5 m/s and average drone altitudes of 44 m. In Figure \ref{fig:convergence_benchmark_schemes}, we analyze the convergence of the \ac{cga}, \ac{genocop}, \ac{sa}, and \ac{drl} benchmark schemes. The plots show that the \ac{cga}, \ac{genocop} and \ac{sa} benchmarks converge to feasible solutions; however, their performance varies, as will be discussed in the next section.\par
\subsection{Relevant Tradeoffs}
\begin{figure}[H]
	\centering
	\ifonecolumn
	\includegraphics[width=4in]{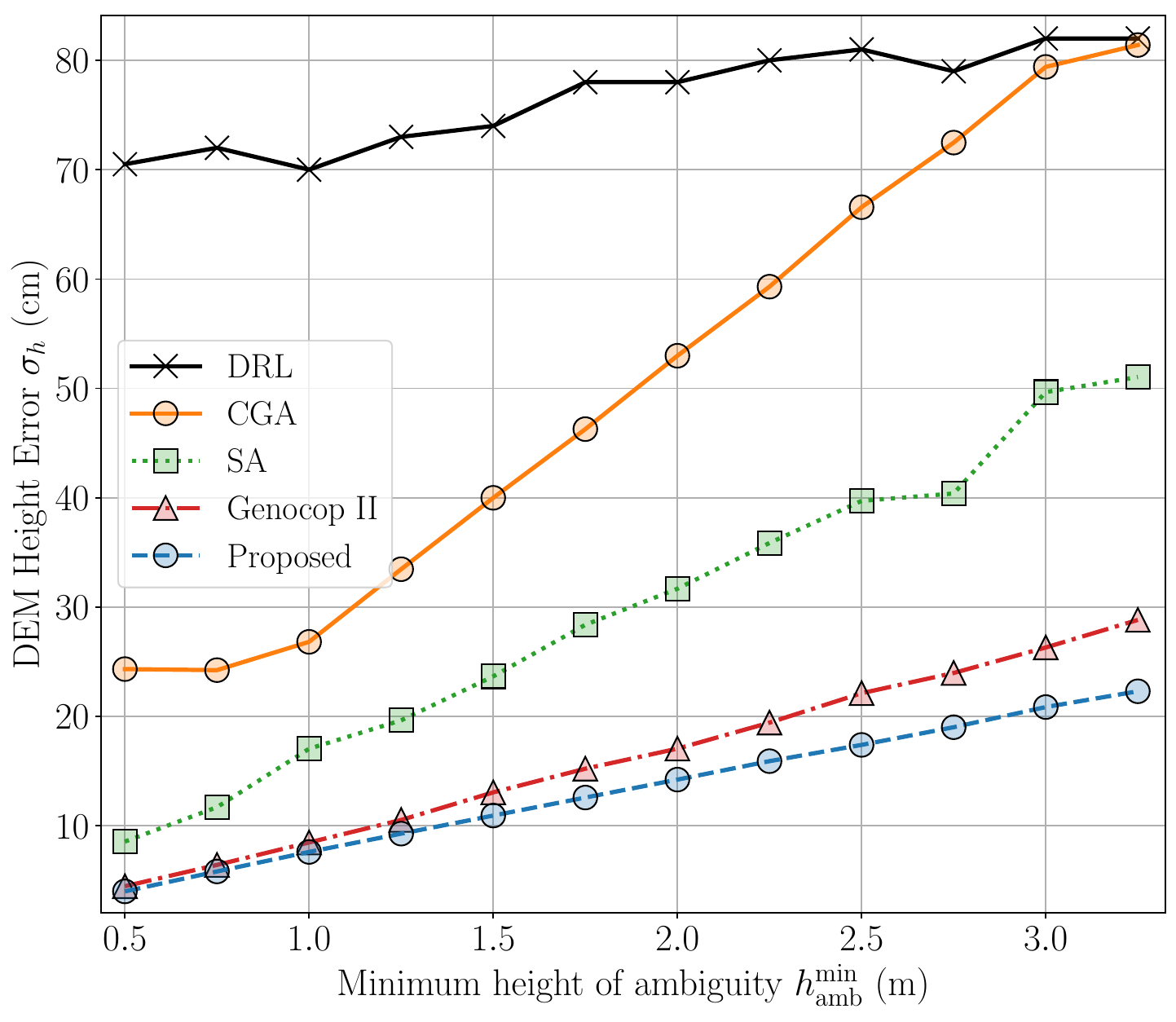}
	\else
	\includegraphics[width=0.9\columnwidth]{figures/error_vs_hmin.pdf}
	\fi
	\caption{Height error of the final \ac{dem} versus the minimum \ac{hoa} for $I=3$.}
	\label{fig:error_vs_hmin}
\end{figure} 
Figure \ref{fig:error_vs_hmin} illustrates the height error of the final \ac{dem} achieved by the proposed solution and all benchmark schemes for different minimum \acp{hoa}, $h_{\mathrm{amb}}^{\mathrm{min}}$. As expected, increasing the \ac{hoa}, which enhances the phase unwrapping robustness, comes at the cost of a larger height error, $\sigma_h$, as shown in (\ref{eq:height_error}). Notably, the \ac{ddpg} algorithm performs the worst among all benchmark schemes. In fact, fine-tuning the hyperparameters of the \ac{ddpg} agent is a complex and sensitive task, requiring substantial training and multiple adjustments. This makes it prone to instability and poor scalability, especially as the problem size (e.g., \ac{uav} network size $I$) grows. As training the agent for each change in system parameters is time-consuming, we omit the results from this benchmark in subsequent plots. The \ac{cga} and \ac{sa} benchmark schemes outperform \ac{drl}, but still fall short of Genocop II and the proposed solution. In fact, these algorithms do not utilize evolutionary methods, limiting their ability to locate high-quality solutions. In contrast, Genocop II yields the best performance among all benchmark schemes, achieving acceptable height error values due to its evolutionary approach and effective handling of constraints. However, this algorithm is computationally intensive, given its complex mutation and crossover operations. Furthermore, Figure \ref{fig:error_vs_hmin} confirms that the proposed solution delivers the best performance, achieving height error values between 5 and 20 cm {for the considered} $h_{\mathrm{amb}}^{\mathrm{min}}$ values. This performance is attributed to the co-evolutionary process, which reduces variable coupling by decomposing the problem and enables collaborative and parallel exploration of different search spaces, leading to superior accuracy in terms of the height error $\sigma_h$.
\begin{figure}[H]
	\centering
	\ifonecolumn
	\includegraphics[width=4in]{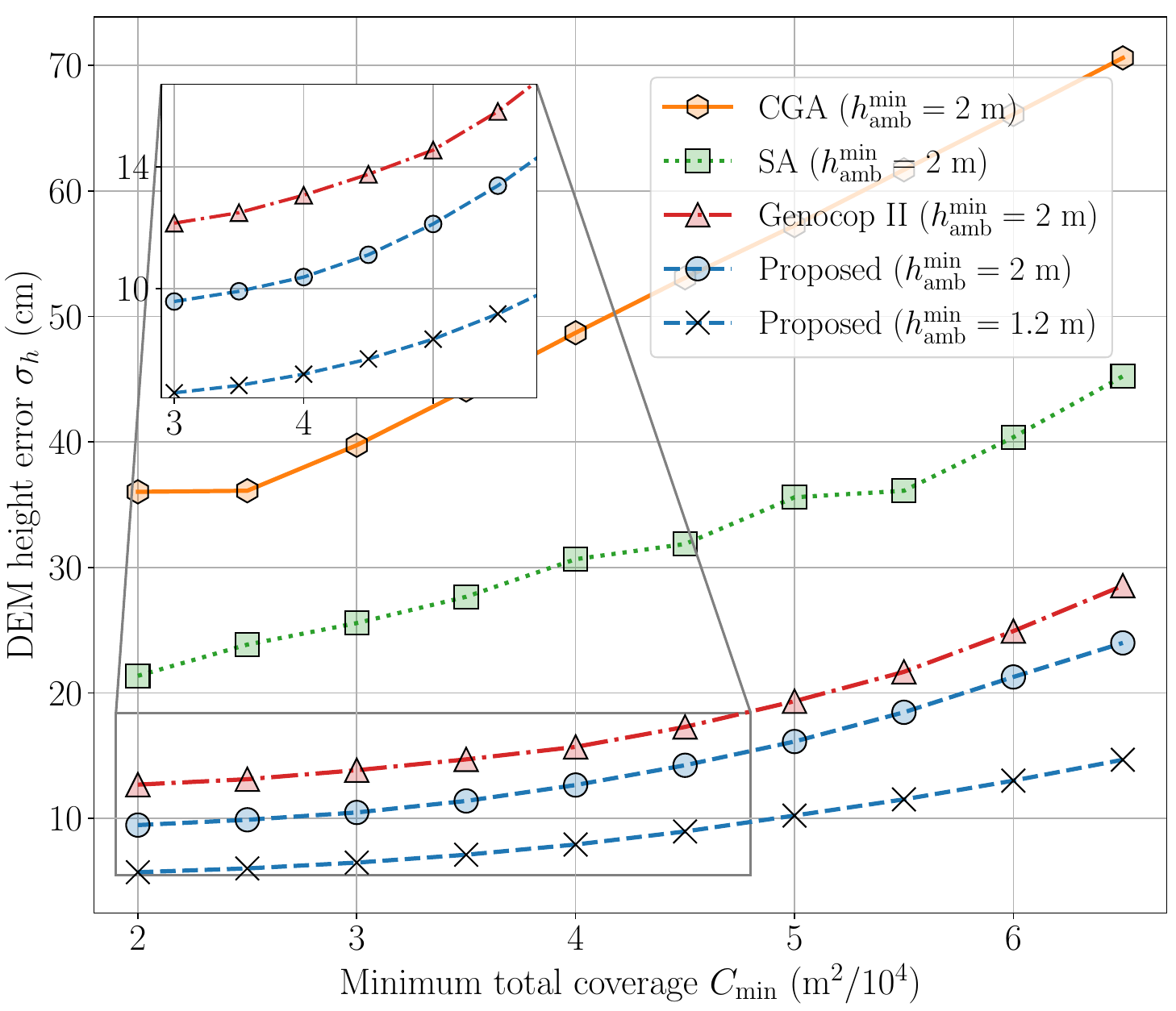}
	\else
	\includegraphics[width=0.9\columnwidth]{figures/error_vs_smin.pdf}
	\fi
	\caption{Height error of the final \ac{dem} versus total coverage requirement for $I=3$.}
	\label{fig:error_vs_smin}
\end{figure} 
Figure \ref{fig:error_vs_smin} illustrates the height error $\sigma_h$ versus the total required \ac{sar} coverage $C_{\mathrm{min}}$ achieved by the proposed solution and benchmark schemes \ac{cga}, \ac{sa}, and \ac{genocop}. The figure reveals that an increase in the coverage requirement leads to degraded sensing performance. This relationship can be attributed to several factors that influence the \ac{uav} trajectories. Specifically, to meet the increased coverage requirement, adjustments are made to both the velocity and radar swath width, which directly impact the height error. Additionally, as the coverage requirement increases, the data rates $R_{\mathrm{min},i}$ for some \acp{uav} $U_i$ also increase. Consequently, some drones adjust their positions to improve their channel quality and {to} satisfy the data offloading constraints. Furthermore, the proposed solution achieves superior performance for lower minimum \ac{hoa} values, which aligns with the trend observed in Figure \ref{fig:error_vs_hmin}. Moreover, the proposed solution outperforms all other benchmark schemes, including \ac{genocop}, which achieves the best performance among the benchmarks. In particular, the proposed solution shows a performance gain of 15\% for $C_{\mathrm{min}} = 6.5 \times 10^4$ m$^2$ and up to 25\% for $C_{\mathrm{min}} = 2 \times 10^4$ m$^2$.\par
	\begin{figure}[H]
	\centering
	\ifonecolumn
	\includegraphics[width=4in]{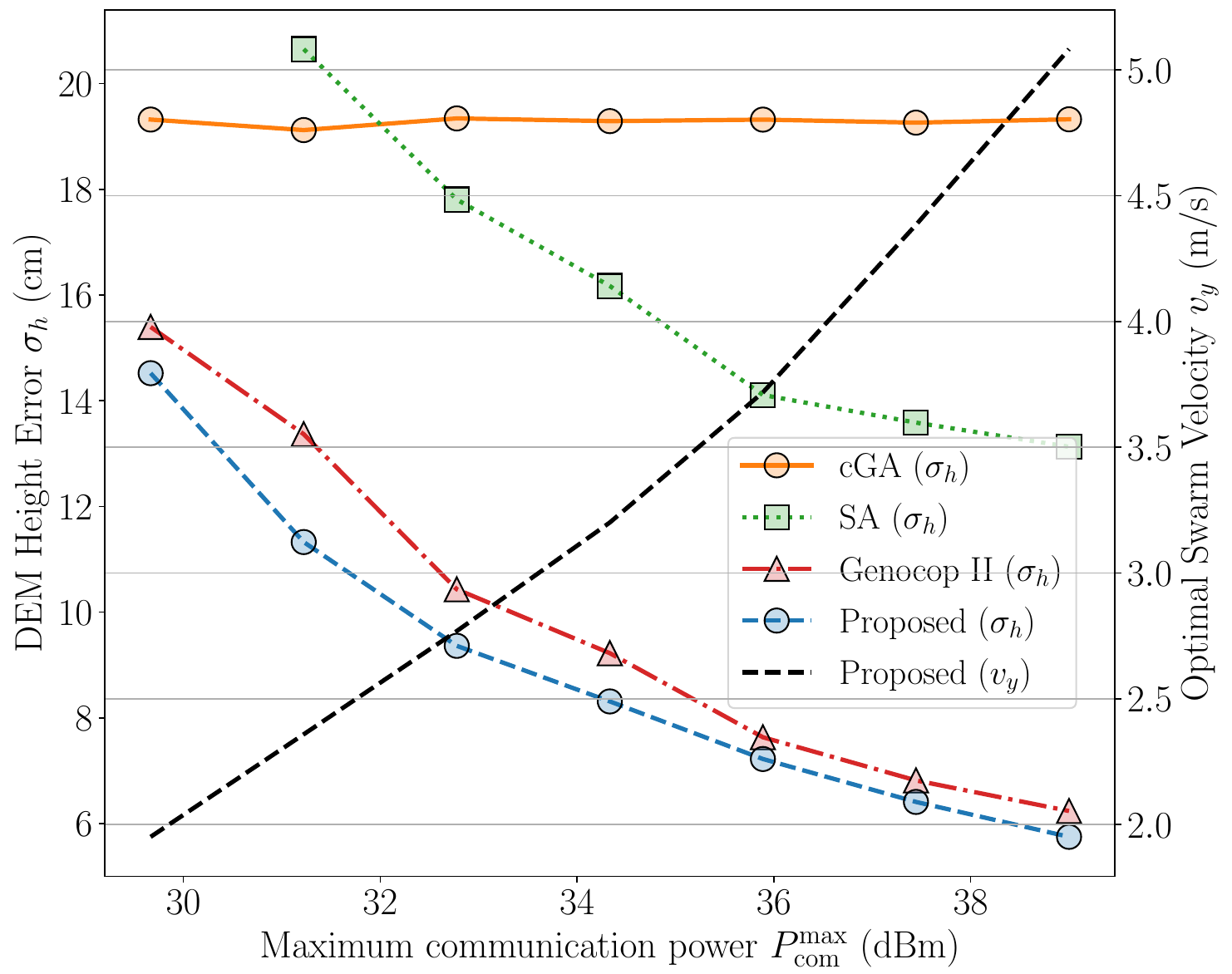}
	\else
	\includegraphics[width=0.9\columnwidth]{figures/error_vs_pcom.pdf}
	\fi
	\caption{Height error of the final \ac{dem} (left axis) and optimal \ac{uav} swarm velocity (right axis) versus the maximum communication power $P_{\mathrm{com}}^{\mathrm{max}}$ for a network size $I=5$. The \ac{sa} scheme fails to find a feasible solution for $P_{\mathrm{com}}^{\mathrm{max}} < 31.5~\text{dBm}$, while no scheme can find a feasible solution for $P_{\mathrm{com}}^{\mathrm{max}} < 29.5~\text{dBm}$.}
	\label{fig:error_vs_pcom}
\end{figure}
In Figure \ref{fig:error_vs_pcom}, we present the height error of the final \ac{dem} as well as the optimal \ac{uav} velocity obtained by the proposed solution, plotted against the maximum communication power used for offloading the radar data. For reduced $P_{\mathrm{com}}^{\mathrm{max}}$ values, the air-to-ground channel capacity deteriorates. To accommodate the data rate requirement in this case, the proposed co-evolutionary algorithm dynamically adjusts the \ac{uav} formation and lowers the swarm velocity to keep the drones in regions with feasible air-to-ground offloading while ensuring the minimum total coverage. This alters the \ac{insar} baselines and results in degraded sensing accuracy, leading to the conclusion that stringent communication requirements reduce sensing precision. For higher $P_{\mathrm{com}}^{\mathrm{max}}$, the communication reach of the drones is improved. The proposed solution then increases the velocity of the swarm to meet the minimum required total coverage and benefits from greater flexibility in placing the drones in the across-track plane, which yields larger baselines and improved accuracy. This figure unveils an interesting trade-off between sensing and communication in the considered system.
\subsection{Optimal UAV Formation} 
\begin{table}[H]
		\centering
\caption{Optimal values of key system parameters. The count denotes the number of parameters optimized \ac{wrt} the number of drones in the network, e.g., $I=5$ altitude values are optimized for a network with $I=5$ drones. For each parameter, the minimum, maximum, and mean values are computed \ac{wrt} the count.}
\label{tab:optimal_parameters}
\begin{adjustbox}{max width=\columnwidth}
		\begin{tabular}{|c|c|ccc|}
			\hline
			\rowcolor[HTML]{EFEFEF} 
			\textbf{Parameter (unit)}  & \textbf{Count} & \multicolumn{1}{c|}{\cellcolor[HTML]{EFEFEF}\textbf{Min}} & \multicolumn{1}{c|}{\cellcolor[HTML]{EFEFEF}\textbf{Max}} & \cellcolor[HTML]{EFEFEF}\textbf{Mean} \\ \hline
			Altitude (m)               & $I=5$            & \multicolumn{1}{c|}{30.8}                                 & \multicolumn{1}{c|}{64.6}                                 & 44                                    \\ \hline
			Cross-range (m)            & $I=5$            & \multicolumn{1}{c|}{-49}                                  & \multicolumn{1}{c|}{-12.6}                                & -24.4                                 \\ \hline
			Swath (m)                  & $I=5$            & \multicolumn{1}{c|}{45.7}                                 & \multicolumn{1}{c|}{96.5}                                 & 62.2                                  \\ \hline
			Baseline (m)               & $\frac{I(I-1)}{2}=10$  & \multicolumn{1}{c|}{4.6}                                  & \multicolumn{1}{c|}{49.5}                                 & 24.6                                  \\ \hline
			Perpendicular baseline (m) & $\frac{I(I-1)}{2}=10$            & \multicolumn{1}{c|}{0.5}                                  & \multicolumn{1}{c|}{6.6}                                  & 3.6                                   \\ \hline
			Velocity (m/s)             & $1$              & \multicolumn{3}{c|}{4.5}                                                                                                                                      \\ \hline
		\end{tabular}%
\end{adjustbox}
\end{table}
In Table \ref{tab:optimal_parameters}, we summarize some of the key system parameters optimized with \textbf{Algorithm} \ref{alg:coevolution}. The results show that the \ac{uav} network adopts a wide range of baselines to achieve optimal performance, with the perpendicular baseline varying from 0.5 m to 6.6 m. The \ac{uav} swarm operates at an  altitude of around 44 m, reaching a maximum altitude of 64.6 m, and moves at 4.5 m/s along the $y$-direction. The table also confirms that the collision-avoidance constraint is satisfied, as the minimum baseline is larger than 4.5 m.
\subsection{Effect of Network Size}
	 \begin{figure}
	 	\centering
	 	\ifonecolumn
	 		\includegraphics[width=4in]{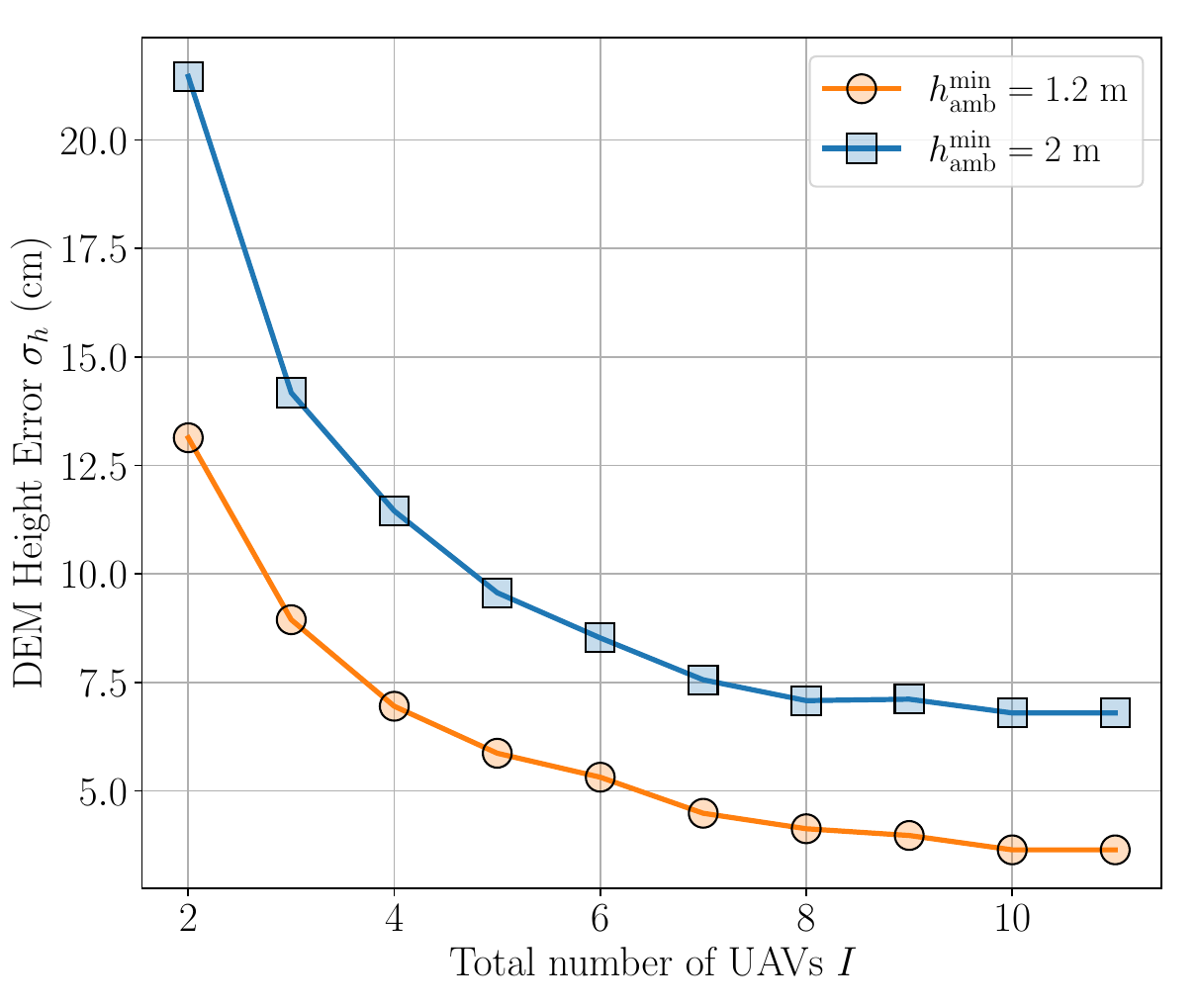}
	 	\else
	 \includegraphics[width=0.9\columnwidth]{figures/error_vs_number_drones.pdf}
	 	\fi
	 		\caption{Height error of the final \ac{dem} versus the total number of drones $I$ for different \acp{hoa}.}
	 		\label{fig:error_vs_number_drones}
	 	\end{figure} 
Figure \ref{fig:error_vs_number_drones} investigates the effect of the total number of drones $I$ on the sensing performance achieved by the proposed co-evolutionary algorithm for different minimum \ac{hoa} values. Notably, the precision enhancement from multi-baseline \ac{insar}, i.e., $I \geq 3$, is substantial when compared to single-baseline \ac{insar} ($I = 2$). In addition, the figure shows an improvement in the sensing performance as the number of drones {increases}, but this improvement progressively diminishes and eventually {saturates} as $I$ grows. In practice, beyond $I = 8$ sensing platforms, there is no significant reduction in height error for both $h_{\mathrm{amb}}^{\mathrm{min}} = 1.2$ m and $h_{\mathrm{amb}}^{\mathrm{min}} = 2$ m. The reduction in height error for a low number of drones is attributed to the averaging of multiple acquisitions, see (\ref{eq:height_error}). When the number of sensors is high, the aggregate contribution of the additional drones becomes negligible, as the height error of the averaged \ac{dem} is primarily dominated by the best individual measurement. Yet, the enhanced baseline diversity introduced by additional \ac{uav}s enables more efficient phase unwrapping by utilizing the smallest baselines for this task, while the larger baselines are employed to improve the height error. Figure \ref{fig:error_vs_number_drones} highlights the potential of \ac{uav} swarms to enhance \ac{sar} sensing capabilities and reveals the optimal network size necessary to achieve a specific sensing performance.
\section{Conclusion} \label{sec:conclusion}
In this paper, we studied \ac{uav}-based multi-baseline \ac{insar} sensing for generating high-precision \acp{dem}, where an arbitrary number of drones are deployed to create independent \acp{dem} of a target area. The final DEM is derived through a weighted averaging technique, thereby improving sensing accuracy. We propose a novel co-evolutionary algorithm that jointly optimizes the drone formation, swarm velocity, and communication resource allocation to minimize the height error of the final DEM. Numerical simulations demonstrate that the proposed algorithm outperforms several {benchmark schemes}, including genetic algorithms, simulated annealing, and deep reinforcement learning, in terms of sensing accuracy. Furthermore, we {unveiled} the potential of coordinated \ac{uav} swarms for achieving centimetric precision for \ac{insar} applications. Future work will focus on further refining the optimization strategy to handle dynamic environmental conditions and enable real-time adaptation using online optimization frameworks.
	\bibliographystyle{IEEEtran}
	\bibliography{biblio}

\begin{thebibliography}{10}
\providecommand{\url}[1]{#1}
\csname url@samestyle\endcsname
\providecommand{\newblock}{\relax}
\providecommand{\bibinfo}[2]{#2}
\providecommand{\BIBentrySTDinterwordspacing}{\spaceskip=0pt\relax}
\providecommand{\BIBentryALTinterwordstretchfactor}{4}
\providecommand{\BIBentryALTinterwordspacing}{\spaceskip=\fontdimen2\font plus
\BIBentryALTinterwordstretchfactor\fontdimen3\font minus
  \fontdimen4\font\relax}
\providecommand{\BIBforeignlanguage}[2]{{%
\expandafter\ifx\csname l@#1\endcsname\relax
\typeout{** WARNING: IEEEtran.bst: No hyphenation pattern has been}%
\typeout{** loaded for the language `#1'. Using the pattern for}%
\typeout{** the default language instead.}%
\else
\language=\csname l@#1\endcsname
\fi
#2}}
\providecommand{\BIBdecl}{\relax}
\BIBdecl

\bibitem{amine5}
M.-A. Lahmeri \emph{et~al.}, ``Sensing accuracy optimization for
  communication-assisted dual-baseline {UAV-InSAR},'' \emph{Proc. Int. Conf.
  Commun.}, 2024.

\bibitem{radar_book}
J.~C. Curlander and R.~N. McDonough, \emph{Synthetic Aperture Radar: Systems
  and Signal Processing}.\hskip 1em plus 0.5em minus 0.4em\relax John Wiley and
  Sons, 1991.

\bibitem{experimental1}
A.~Grathwohl \emph{et~al.}, ``Detection of objects below uneven surfaces with a
  {UAV}-based {GPSAR},'' \emph{IEEE Trans. Geosci. Remote Sens.}, vol.~61, pp.
  1--13, 2023.

\bibitem{experimental2}
E.~Schreiber, A.~Heinzel, M.~Peichl, M.~Engel, and W.~Wiesbeck, ``Advanced
  buried object detection by multichannel, {UAV}/drone carried synthetic
  aperture radar,'' in \emph{Proc. Euro. Conf. Antennas Propag.}, 2019, pp.
  1--5.

\bibitem{insar_exp1}
O.~Frey and C.~L. Werner, ``{UAV}-borne repeat-pass {SAR} interferometry and
  {SAR} tomography with a compact {L}-band {SAR} system,'' in \emph{Proc. Eur.
  Conf. Synth. Aperture Radar}, 2021, pp. 1--4.

\bibitem{insar_exp2}
V.~Mustieles-Perez \emph{et~al.}, ``Towards {UAV}-based ultra-wideband
  multi-baseline {SAR} interferometry,'' in \emph{Proc. Euro. Radar Conf.},
  2023, pp. 233--236.

\bibitem{insar_exp3}
------, ``Experimental demonstration of {UAV}-based ultra-wideband
  multi-baseline {SAR} interferometry,'' in \emph{Proc. Euro. Conf. Synth.
  Aperture Radar}, 2024, pp. 1156--1161.

\bibitem{climate_change}
I.~Ullmann \emph{et~al.}, ``Towards detecting climate change effects with
  {UAV}-borne imaging radars,'' \emph{IEEE J. Microw.}, vol.~4, no.~4, pp.
  881--893, 2024.

\bibitem{optimization1}
Z.~Guan, Z.~Sun, J.~Wu, and J.~Yang, ``Resource allocation and optimization of
  multi-{UAV} {SAR} system,'' in \emph{Proc. IEEE Radar Conf.}, 2023, pp. 1--5.

\bibitem{amine2}
M.-A. Lahmeri \emph{et~al.}, ``Trajectory and resource optimization for {UAV}
  synthetic aperture radar,'' in \emph{Proc. IEEE Global Commun. Conf.}, 2022,
  pp. 897--903.

\bibitem{amine4}
M.-A. Lahmeri, V.~Mustieles-Pérez, M.~Vossiek, G.~Krieger, and R.~Schober,
  ``{UAV} formation and resource allocation optimization for
  communication-assisted {3D} {InSAR} sensing,'' \emph{IEEE Trans. Commun.},
  vol.~73, no.~8, pp. 5788--5804, 2025.

\bibitem{sun2}
Z.~Sun \emph{et~al.}, ``Trajectory optimization for maneuvering platform
  bistatic {SAR} with geosynchronous illuminator,'' \emph{IEEE Trans. Geosci.
  Remote Sens.}, vol.~62, 2024, {Art.} No. 5203715.

\bibitem{cramer}
P.~Rosen \emph{et~al.}, ``Synthetic aperture radar interferometry,''
  \emph{Proc. IEEE}, vol.~88, no.~3, pp. 333--382, 2000.

\bibitem{insar1}
R.~Bamler and P.~Hartl, ``Synthetic aperture radar interferometry,''
  \emph{Inverse Probl.}, vol.~14, no.~4, 1998.

\bibitem{added_krieger}
P.~Rizzoli \emph{et~al.}, ``Generation and performance assessment of the global
  {TanDEM-X} digital elevation model,'' \emph{ISPRS J. Photogramm. Remote
  Sens.}, vol. 132, pp. 119--139, 2017.

\bibitem{coherence1}
M.~Martone \emph{et~al.}, ``Coherence evaluation of {TanDEM-X} interferometric
  data,'' \emph{ISPRS J. Photogramm. Remote Sens.}, vol.~73, pp. 21--29, 2012.

\bibitem{optimization2}
J.~Drozdowicz and P.~Samczynski, ``Drone-based {3D} synthetic aperture radar
  imaging with trajectory optimization,'' \emph{Sensors}, vol.~22, no.~18, p.
  6990, 2022.

\bibitem{amine3}
M.-A. Lahmeri \emph{et~al.}, ``{UAV} formation optimization for
  communication-assisted {InSAR} sensing,'' in \emph{Proc. IEEE Int. Conf.
  Commun.}, 2024, pp. 3913--3918.

\bibitem{added_krieger2}
G.~Krieger and A.~Moreira, ``Multistatic sar satellite formations: potentials
  and challenges,'' in \emph{Proc. IEEE Int. Geosci. and Remote Sens. Symp.},
  vol.~4, 2005, pp. 2680--2684.

\bibitem{added_krieger3}
M.~Lachaise, T.~Fritz, and R.~Bamler, ``The dual-baseline phase unwrapping
  correction framework for the {TanDEM-X} mission part 1: Theoretical
  description and algorithms,'' \emph{IEEE Trans. Geosci. Remote Sens.},
  vol.~56, no.~2, pp. 780--798, 2018.

\bibitem{fusion_sigma_h}
A.~Gruber \emph{et~al.}, ``The {TanDEM-X DEM} mosaicking: Fusion of multiple
  acquisitions using {InSAR} quality parameters,'' \emph{IEEE J. Sel. Top.
  Appl. Earth Obs. Remote Sens.}, vol.~9, no.~3, pp. 1047--1057, 2016.

\bibitem{book1}
M.~I. Skolnik \emph{et~al.}, \emph{Introduction to Radar Systems}.\hskip 1em
  plus 0.5em minus 0.4em\relax McGraw-hill New York, 1980, vol.~3.

\bibitem{snr_equation}
G.~Krieger \emph{et~al.}, ``{TanDEM-X}: A satellite formation for
  high-resolution {SAR} interferometry,'' \emph{IEEE Trans. Geosci. Remote
  Sens.}, vol.~45, no.~11, pp. 3317--3341, 2007.

\bibitem{victor2}
V.~Mustieles-Perez \emph{et~al.}, ``New insights into wideband synthetic
  aperture radar interferometry,'' \emph{IEEE Geosci. Remote Sens. Lett.},
  vol.~21, pp. 1--5, 2024.

\bibitem{propulsion}
Y.~Zeng, J.~Xu, and R.~Zhang, ``Energy minimization for wireless communication
  with rotary-wing {UAV},'' \emph{IEEE Trans. Wireless Commun.}, vol.~18,
  no.~4, pp. 2329--2345, 2019.

\bibitem{evolutionary_algorithms2}
T.~Bäck and H.-P. Schwefel, ``An overview of evolutionary algorithms for
  parameter optimization,'' \emph{Evol. Comput.}, vol.~1, no.~1, pp. 1--23,
  1993.

\bibitem{evolutionary_algorithms}
P.~A. Vikhar, ``Evolutionary algorithms: A critical review and its future
  prospects,'' in \emph{Proc. IEEE Int. Conf. Global Trends Signal Process.,
  Inf. Comput. Commun.}, 2016, pp. 261--265.

\bibitem{coevolution_survey}
L.~Miguel~Antonio and C.~A. Coello, ``Coevolutionary multiobjective
  evolutionary algorithms: Survey of the state-of-the-art,'' \emph{IEEE Trans.
  Evol. Comput.}, vol.~22, no.~6, pp. 851--865, 2018.

\bibitem{coevolution_largescale2}
M.~N. Omidvar, X.~Li, Z.~Yang, and X.~Yao, ``Cooperative co-evolution for large
  scale optimization through more frequent random grouping,'' in \emph{Proc.
  IEEE Congr. Evol. Comput.}, 2010, pp. 1--8.

\bibitem{coevolution_pso}
Q.~He and L.~Wang, ``An effective co-evolutionary particle swarm optimization
  for constrained engineering design problems,'' \emph{Eng. Appl. Artif.
  Intell.}, vol.~20, no.~1, pp. 89--99, 2007.

\bibitem{coevolution_largescale1}
Z.~Yang, K.~Tang, and X.~Yao, ``Large scale evolutionary optimization using
  cooperative coevolution,'' \emph{Inf. Sci.}, vol. 178, no.~15, pp.
  2985--2999, 2008.

\bibitem{coevolution_dynamic}
C.-K. Goh and K.~C. Tan, ``A competitive-cooperative coevolutionary paradigm
  for dynamic multiobjective optimization,'' \emph{IEEE Trans. Evol. Comput.},
  vol.~13, no.~1, pp. 103--127, 2009.

\bibitem{coevolution_multiobjective}
Y.~Tian \emph{et~al.}, ``A coevolutionary framework for constrained
  multiobjective optimization problems,'' \emph{IEEE Trans. Evol. Comput.},
  vol.~25, no.~1, pp. 102--116, 2021.

\bibitem{coevolution_butterfly}
P.~R. Ehrlich and P.~H. Raven, ``Butterflies and plants: a study in
  coevolution,'' \emph{Evolution}, pp. 586--608, 1964.

\bibitem{non_parameterized_pso}
K.~Deb, ``An efficient constraint handling method for genetic algorithms,''
  \emph{Comput. Methods Appl. Mech. Eng.}, vol. 186, no. 2-4, pp. 311--338,
  2000.

\bibitem{pso}
J.~Kennedy and R.~Eberhart, ``Particle swarm optimization,'' in \emph{Proc.
  Int. Conf. Neural Netw.}, vol.~4, 1995, pp. 1942--1948.

\bibitem{antenna_array}
Q.~Xu \emph{et~al.}, ``On formulating and designing antenna arrays by
  evolutionary algorithms,'' \emph{IEEE Trans. Antennas Propag.}, vol.~69,
  no.~2, pp. 1118--1129, 2021.

\bibitem{path_planning}
V.~Roberge, M.~Tarbouchi, and G.~Labonte, ``Comparison of parallel genetic
  algorithm and particle swarm optimization for real-time {UAV} path
  planning,'' \emph{IEEE Trans. Ind. Informat.}, vol.~9, no.~1, pp. 132--141,
  2013.

\bibitem{scheduling}
K.~Gao \emph{et~al.}, ``A review on swarm intelligence and evolutionary
  algorithms for solving flexible job shop scheduling problems,''
  \emph{IEEE/CAA J. Automatica Sinica}, vol.~6, no.~4, pp. 904--916, 2019.

\bibitem{performance_pso}
Y.~Shi and R.~Eberhart, ``Empirical study of particle swarm optimization,'' in
  \emph{Proc. 1999 Congr. Evol. Comput.}, vol.~3, 1999, pp. 1945--1950.

\bibitem{boundary}
J.~Robinson and Y.~Rahmat-Samii, ``Particle swarm optimization in
  electromagnetics,'' \emph{IEEE Trans. Antennas Propag.}, vol.~52, no.~2, pp.
  397--407, 2004.

\bibitem{constrained_pso}
X.~Hu and R.~Eberhart, ``Solving constrained nonlinear optimization problems
  with particle swarm optimization,'' in \emph{Proc. 6th World Multiconf.
  Systemics, Cybernetics Inform.}, 2002.

\bibitem{cga}
R.~L. Haupt and S.~E. Haupt, \emph{Practical Genetic Algorithms}.\hskip 1em
  plus 0.5em minus 0.4em\relax John Wiley \& Sons, 2004.

\bibitem{genocopII}
Z.~Michalewicz and N.~Attia, ``Evolutionary optimization of constrained
  problems,'' in \emph{Proc. 3rd Annu. Conf. Evol. Programming}.\hskip 1em plus
  0.5em minus 0.4em\relax World Scientific Publishing, 1994, pp. 98--108.

\bibitem{genocopI}
Z.~Michalewicz and C.~Z. Janikow, ``{GENOCOP}: a genetic algorithm for
  numerical optimization problems with linear constraints,'' \emph{Commun.
  ACM}, vol.~39, no.~12, pp. 175--185, 1996.

\bibitem{sa}
D.~Delahaye, S.~Chaimatanan, and M.~Mongeau, ``Simulated annealing: From basics
  to applications,'' in \emph{Handbook of Metaheuristics}.\hskip 1em plus 0.5em
  minus 0.4em\relax {Springer}, vol. 272, pp. 1--35.

\bibitem{ddpg}
T.~P. Lillicrap \emph{et~al.}, ``Continuous control with deep reinforcement
  learning,'' \emph{arXiv preprint arXiv:1509.02971}, 2015.

\end{thebibliography}
	
	%%%%%%%%%%%%%%%%%%%%%%%%%%%%%%%%%%%%%%%%%%%%%%%%%%%%%%%%%%%%%%%%%
\end{document}